\documentclass{article}
\PassOptionsToPackage{numbers,compress}{natbib}
\usepackage[main, preprint]{neurips_2026}
\usepackage[utf8]{inputenc}
\usepackage[T1]{fontenc}
\usepackage{hyperref}
\usepackage{url}
\usepackage{booktabs}
\usepackage{amsfonts}
\usepackage{amsmath}
\usepackage{amssymb}
\usepackage{nicefrac}
\usepackage{amsthm}
\newtheorem{proposition}{Proposition}
\usepackage{microtype}
\usepackage{xcolor}
\usepackage{enumitem}
\usepackage{multirow}
\usepackage{array}
\usepackage{graphicx}
\hypersetup{hidelinks}
\usepackage{cleveref}
\usepackage{mathtools}
 
\title{Neuronal Identity as an Organizational Basis for Analyzing Neural Population Dynamics}

\author{%
  Seungjae Han$^{1*}$
  \quad  Joshua Yedam You$^{1*}$
  \quad  Soi Kim$^{1*}$
  \quad  Minho Eom$^{1}$
  \quad  Jinho Park$^{1}$ \\
  \bfseries \quad  Seunghee Park$^{1}$
  \quad  Byeongwook Lee$^{2,3,4\dagger}$
  \quad  Young-Gyu Yoon$^{1,5,6\dagger}$ \\
  \\
  \mdseries $^{1}$School of Electrical Engineering, KAIST, Daejeon, Republic of Korea \\
  \mdseries $^{2}$Department of Brain and Cognitive Sciences, KAIST, Daejeon, Republic of Korea \\
  \mdseries $^{3}$Graduate School of Data Science, KAIST, Daejeon, Republic of Korea \\
  \mdseries $^{4}$Department of AX, KAIST, Daejeon, Republic of Korea \\
  \mdseries $^{5}$Department of Semiconductor System Engineering, KAIST, Daejeon, Republic of Korea \\
  \mdseries $^{6}$KAIST Institute for Human Augmentation Convergence, KAIST, Daejeon, Republic of Korea \\
  \mdseries $^{*}$Co-first authors \quad $^{\dagger}$Co-corresponding authors: \texttt{\{bwlee, ygyoon\}@kaist.ac.kr}\\
}

\begin{document}

\renewcommand{\thefootnote}{}
\footnotetext{Project page and code: \url{https://nicalab.github.io/identity-group-routing}}

\maketitle

\begin{abstract}
Individual neurons exhibit diverse molecular, morphological, and electrophysiological properties that shape their distinct roles in neural circuits. While these multimodal features are largely shared across individuals, they remain an underutilized source of information in analyzing neural population dynamics. We posit that these shared biological profiles provide a common reference for relating neural dynamics across subjects and recording sessions. Our framework leverages this structure by routing neurons into latent groups based on their biological profiles and analyzing population dynamics through group-level representations. Because these assignments are anchored in biological properties that are consistent across subjects, the learned organization supports transfer to unseen animals without retraining or neuron-wise correspondence. We evaluate our framework on whole-brain calcium imaging in larval zebrafish and large-scale electrophysiology in the mouse visual cortex. Across both species and recording modalities, the framework delivers state-of-the-art performance in both neural decoding and manifold analysis: it achieves superior decoding of behavior and stimulus labels while better preserving the temporal consistency and local geometric structure of latent trajectories on unseen data. Our experiments demonstrate that not only does providing biological features improve baseline performance, but our proposed framework achieves even stronger results by using these properties to explicitly structure the population rather than treating them as auxiliary inputs. Together, these results suggest that organizing neural populations by their multimodal biological features offers a promising framework for generalizing population dynamics across subjects and recording sessions.
\end{abstract}

\section{Introduction}
\label{sec:intro}

Recent advances in large-scale neural recording and analysis have opened new opportunities for studying how neural populations give rise to perception, action, and cognition, and have spurred a growing body of deep learning approaches for analyzing population dynamics~\cite{urai2022large,stringer2024analysis,schneider2023cebra,gosztolai2025marble,han2024pixelstoinfo,mathis2024decoding,you2025preserving,yoon2026event}. Yet existing approaches still struggle with generalization across individuals. Every recording captures a different set of neurons, and activity patterns shift with recording conditions, internal states, and individual variability~\cite{azabou2023poyo,arora2025know}. Because the recorded populations themselves differ across animals, there is no guaranteed correspondence between the neural populations, presenting a significant challenge for the direct transfer of learned representations to unseen individuals.

The multimodal biological features of individual neurons, or neuronal identity, represent a complementary source of information that has remained largely underutilized in population dynamics analysis. Neurons can be characterized along multiple biological axes, including transcriptomic, regional, and electrophysiological features, which techniques like Patch-seq integrate at single-cell resolution to link molecular identity with distinct functional roles in neural circuits~\cite{gouwens2020integrated}. For example, in the mouse visual cortex, Pvalb, Sst, and Vip interneurons play specialized, non-redundant roles in gain control, novelty detection, and state-dependent modulation~\cite{pfeffer2013inhibition}, illustrating how neuronal identity shapes individual contributions to circuit computation~\cite{arkhipov2025multimodal_circuits}. Large-scale transcriptomic studies further demonstrate that neurons close in molecular profile participate in related circuit computations~\cite{zeisel2018molecular, tasic2018shared, yao2021taxonomy, biccn2021motorcortex}. Importantly, these features are largely shared across individuals of the same species~\cite{tasic2018shared}, suggesting they may offer a more stable basis for analyzing population dynamics across individuals.

We propose to exploit this structure through an identity-guided, permutation-invariant framework for analyzing neural population dynamics. Neurons are routed into latent groups based on their biological profiles, including transcriptomic, regional, and electrophysiological features, and population dynamics are analyzed at the group level through shared parameters conditioned on these features. Because neuron ordering in a neural recording dataset is arbitrary and varies across recordings, the framework should also be permutation-invariant. Standard permutation-invariant architectures such as DeepSets~\cite{zaheer2017deep} achieve this by treating all neurons as interchangeable units through shared element-wise transformations followed by symmetric global aggregation.
Our framework instead maintains permutation invariance by encoding each neuron's identity and activity through shared neuron-wise networks, then routing these representations into latent groups via assignment-weighted aggregation. Because these group assignments are determined by biological properties that are stable across individuals, the learned organization can be applied to unseen animals without retraining or neuron-wise correspondence.

This design contrasts with existing approaches for analyzing neural population dynamics. Recent large-scale models, such as POYO~\cite{azabou2023poyo} and NEDS~\cite{zhang2025neds}, achieve cross-session transfer by learning session-specific embeddings and relying on statistical regularities in observed dynamical patterns to bridge recordings. Neural manifold analysis methods such as CEBRA~\cite{schneider2023cebra} and MARBLE~\cite{gosztolai2025marble} learn low-dimensional embeddings of population dynamics that reveal how neural circuits represent and transform information, providing an approach for relating population activity to behavior and perception. However, because these approaches rely primarily on activity patterns, they often struggle to generalize to new animals or sessions. Our framework instead grounds each neuron in a common biological coordinate, thereby inducing a shared group structure that supports transfer to held-out animals and sessions without any neuron-wise correspondence.

We evaluate our framework on two large-scale datasets spanning distinct species, recording modalities, and available biological features: a whole-brain larval zebrafish calcium imaging dataset that jointly acquires neural activity and single-cell transcriptomics~\cite{marquezlegorreta2026warp}, and a large-scale mouse visual cortex Neuropixels electrophysiology dataset~\cite{siegle2021survey}. Across both settings, our framework achieves state-of-the-art cross-session behavior and stimulus decoding when applied directly to held-out recordings without retraining or neuron-wise correspondence. While supplying baselines with these same neuronal identity features improves their performance, our identity-guided routing framework consistently outperforms all baselines. These results demonstrate both the inherent value of multimodal biological features of individual neurons and the critical importance of explicitly using them to organize population dynamics rather than just treating them as auxiliary inputs. Our contributions are as follows:
\begin{itemize}
    \item We introduce, to our knowledge, the first neural population analysis framework that leverages multimodal biological features of individual neurons to support direct transfer across animals and recording sessions without retraining or neuron-wise correspondence.
    
    \item We propose an identity-guided routing method that uses multimodal biological features as the organizing basis to assign neurons to latent groups, thereby analyzing population dynamics at the group level through shared parameters.
    
    \item Across whole-brain zebrafish calcium imaging recordings and mouse visual cortex electrophysiology, our framework achieves state-of-the-art neural decoding of behavior and stimulus labels and enables robust neural manifold analysis by preserving the temporal and local geometric structure of latent trajectories in held-out animals.
\end{itemize}

\section{Method}

We seek a framework for analyzing neural population dynamics that transfers across sessions without neuron-wise correspondence. The core challenge is that neurons cannot be matched across animals by recording index because neuron order in a neural recording dataset is arbitrary and changes across animals and sessions. 
Neuronal identity provides a shared description of each neuron across individuals. Neurons with high overlap in neuronal identity tend to support similar circuit roles even across individuals, and can serve as a stable basis for organizing the population. This leads to the following two properties: 
\begin{itemize}
    \item \textbf{Permutation invariance.} The model output should be independent of how recorded neurons are ordered.
    \item \textbf{Identity-based organization.} Neurons with shared properties should be considered together, reflecting the group structure observed in biological neural circuits.
\end{itemize}
Permutation-invariant models treat the neural population as an unordered set, requiring aggregation across neurons in a way that does not depend on their ordering~\cite{lee2019set,zaheer2017deep}. A global sum or mean over all neuron representations achieves this, but treats all neurons identically regardless of their neuronal identity, discarding the organizational structure that identity encodes. This motivates an aggregation organized by neuronal identity, where neurons with overlapping identity profiles are pooled into a shared group representation, contributing to population dynamics through a common functional organization. We therefore analyze population dynamics through two components:
\begin{enumerate}
    \item \textbf{Identity-guided routing.} Each neuron is assigned to a latent group on the basis of its neuronal identity features.
    \item \textbf{Group-level dynamics analysis.} 
    Neural activity is aggregated into latent group representations, which are then processed to obtain the final population state.

\end{enumerate}
Because group assignments are determined by neuronal identity, which is shared across individuals, the learned organizational structure can be applied across subjects and sessions without retraining or neuron-wise correspondence. Section~\ref{sec:framework} specifies each stage in detail.
 
\subsection{Framework}
\label{sec:framework}

We represent the neural population activity at a given time point $t$ as a collection of identity--activity pairs
\[
    \mathcal{S}
    =
    \{(\mathbf{f}_i, \mathbf{x}_i)\}_{i=1}^{N},
\]
where \(N\) denotes the number of recorded neurons in the corresponding session. For the $i$-th neuron, 
\(\mathbf{f}_i \in \mathbb{R}^{F}\) is a neuronal identity vector encoding neuron-specific information, and 
\(\mathbf{x}_i \in \mathbb{R}^{W}\) is a segment of its neural activity centered at time $t$. The goal is to learn a mapping 
\begin{equation}
    \mathbf{z}
    =
    \mathcal{F}(\mathcal{S})
    \in
    \mathbb{R}^{D},
    \label{eq:mapping}
\end{equation}
that produces a population state representation $\mathbf{z}$ transferable across sessions and subjects without neuron-wise correspondence. Since neuron ordering is arbitrary and differs across recordings, $\mathcal{F}$ must satisfy permutation invariance: for every permutation $\pi$ of $\{1,\dots,N\}$,

\begin{equation}
    \mathcal{F}\!\bigl(\{(\mathbf{f}_{\pi(i)},\mathbf{x}_{\pi(i)})\}_{i=1}^N\bigr)
    \;=\;
    \mathcal{F}\!\bigl(\{(\mathbf{f}_i,\mathbf{x}_i)\}_{i=1}^N\bigr).
    \label{eq:pi-inv}
\end{equation}
We now formalize each stage of $\mathcal{F}$ in turn and establish that it satisfies permutation invariance.

\paragraph{Identity-Guided Routing.}
 
The first stage maps each neuron's identity features into a soft group membership vector. Neuronal identity comprises $M$ modalities $\mathbf{f}_i = (\mathbf{f}_i^{(1)},\dots,\mathbf{f}_i^{(M)})$ with $\mathbf{f}_i^{(m)} \in \mathbb{R}^{F_m}$, which are heterogeneous in type and measurement scale. Direct concatenation of raw modality features would allow modalities with larger magnitudes to dominate downstream computation irrespective of biological relevance. Each modality is encoded by a dedicated MLP $\phi_m \colon \mathbb{R}^{F_m} \to \mathbb{R}^{d_m}$, and the resulting embeddings are layer-normalized and concatenated before being projected by a shared MLP $\psi \colon \mathbb{R}^{\sum_m d_m} \to \mathbb{R}^{d_s}$:
\begin{equation}
    \mathbf{s}_i
    \;=\;
    \psi\!\Bigl(
        \bigl[
            \mathrm{LN}\!\bigl(\phi_1(\mathbf{f}_i^{(1)})\bigr);\;
            \cdots;\;
            \mathrm{LN}\!\bigl(\phi_M(\mathbf{f}_i^{(M)})\bigr)
        \bigr]
    \Bigr)
    \;\in\; \mathbb{R}^{d_s}.
    \label{eq:identity-enc}
\end{equation}

Given the identity embedding \(\mathbf{s}_i\), neuron \(i\) is assigned a soft membership vector 
\(\mathbf{a}_i \in \Delta^{K-1}\) over \(K\) latent groups, allowing each neuron to partially contribute to multiple groups in proportion to its similarity in the learned identity embedding space:
\begin{equation}
    \mathbf{a}_i
    =
    \sigma\!\left(
        \frac{\mathbf{W}_{\mathrm{route}}\mathbf{s}_i}{\tau}
    \right),
    \qquad
    \mathbf{W}_{\mathrm{route}} \in \mathbb{R}^{K \times d_s},
    \qquad
    \tau > 0,
    \label{eq:routing}
\end{equation}
where \(\sigma:\mathbb{R}^{K}\to\Delta^{K-1}\) denotes the softmax operator.

\paragraph{Group-level dynamics analysis.}

Using the identity-guided routing weights defined in Eq.~\eqref{eq:routing}, we aggregate neuron-level activity embeddings into latent group-level representations. Each neuron's activity window \(\mathbf{x}_{i} \in \mathbb{R}^{W}\) is first encoded by a shared dynamics encoder $\rho : \mathbb{R}^{W} \to \mathbb{R}^{d_r}$, which is applied independently to every neuron:
\begin{equation}
    \mathbf{v}_{i}
    =
    \rho(\mathbf{x}_{i})
    \in
    \mathbb{R}^{d_r}.
    \label{eq:dynamics-enc}
\end{equation}
Because the same encoder \(\rho\) is used for all neurons, the activity representation does not depend on the arbitrary neuron index. For each latent group \(k \in \{1,\dots,K\}\), we form a group-level activity representation by taking a routing-weighted aggregate of the neuron-level activity embeddings:
\begin{equation}
    \mathbf{c}_{k}
    =
    \frac{1}{\sqrt{\max(m_k,\epsilon)}}
    \sum_{i=1}^{N}
    a_{ik}\mathbf{v}_{i},
    \qquad
    m_k
    =
    \sum_{i=1}^{N} a_{ik}.
    \label{eq:aggregation}
\end{equation}
Here, \(a_{ik}\) denotes the contribution of neuron \(i\) to latent group \(k\), while the normalization by \(\sqrt{\max(m_k,\epsilon)}\) stabilizes the scale of group representations across groups, with \(\epsilon=10^{-6}\) included for numerical stability under sparse routing.
Stacking the $K$ group representations yields the group matrix $\mathbf{C} \in \mathbb{R}^{K \times d_r}$.
To model interactions across latent groups, \(\mathbf{C}\) is processed by a group interaction module \(\Gamma:\mathbb{R}^{K\times d_r}\rightarrow \mathbb{R}^{D}\), operating along the group axis:
\begin{equation}
    \mathbf{z}
    =
    \Gamma(\mathbf{C})
    \in
    \mathbb{R}^{D}.
    \label{eq:group-interaction}
\end{equation}
The resulting population state \(\mathbf{z}\) can be used for behavioral decoding, manifold learning, or other downstream objectives, without modifying the upstream identity-guided routing and group-level aggregation framework.

The construction is permutation-invariant by design. Since the identity encoder, dynamics encoder, and routing map are shared across neurons, any permutation of the input neurons only permutes the corresponding neuron-level embeddings and membership weights. The group representations are then formed by summing over the neuron index, so they are unchanged by such a permutation. Consequently, the group matrix \(\mathbf{C}\) and the final population state \(\mathbf{z}=\Gamma(\mathbf{C})\) satisfy the invariance requirement in Eq.~\eqref{eq:pi-inv}. A formal proof of the permutation invariance is provided in Appendix~\ref{app:perm-inv}.

\begin{figure}[t]
    \centering
    \centerline{\includegraphics[width=\columnwidth]{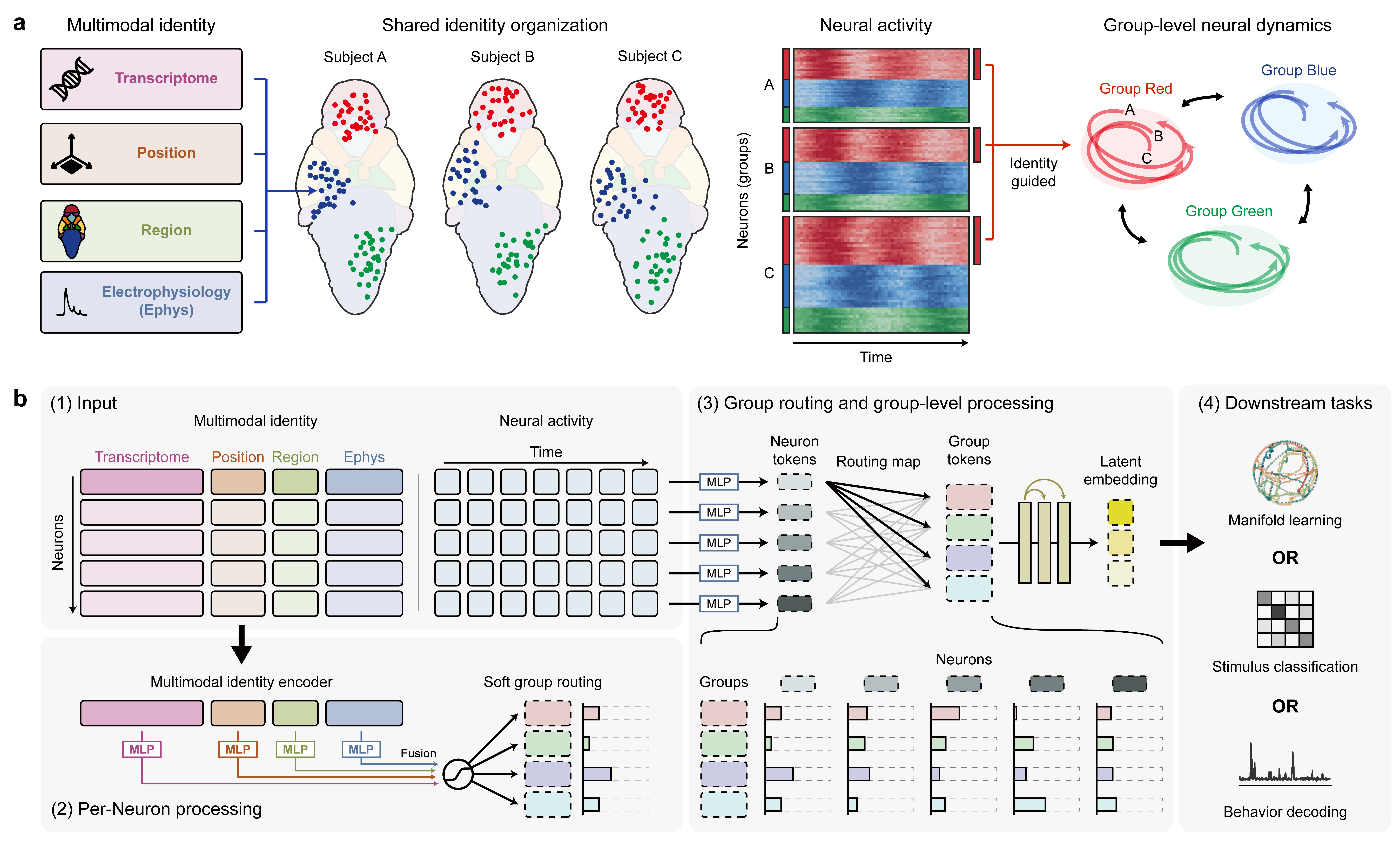}}
    \caption{Overview of identity-guided neural population analysis.
(a) Conceptual overview. Multimodal biological features organize neurons into corresponding groups across sessions, even when exact individual neurons cannot be matched across subjects. Population dynamics are described at the level of these identity-defined groups and their interactions.
(b) Identity-guided routing framework. Identity-related features are encoded to softly route neuron-level activity tokens into latent group tokens, which are processed into a population embedding. Manifold learning, stimulus classification, and behavior decoding are evaluated as separate downstream objectives.
}
    \label{fig:temp}
\end{figure}

\section{Experiments}
\label{sec:experiments}
We evaluate on two datasets spanning different species, recording modalities, and degrees of richness of available neuronal identity information (Section 3.1). We compare against three decoding baselines and four manifold methods (Section 3.2), and assess performance under direct transfer to held-out animals and sessions without retraining using decoding and geometric metrics (Section 3.3).

\subsection{Datasets and Experimental Setup}
\label{sec:datasets}
\paragraph{Whole-brain Neuronal Activity and RNA Profiling 
(WARP).}
The WARP dataset~\cite{marquezlegorreta2026warp} provides whole-brain calcium imaging and single-cell transcriptomics from three larval zebrafish at 6 days post-fertilization, each with approximately 77,000 neurons registered to the mapzebrain atlas~\cite{kunst2019mapzebrain}. Recordings span approximately 70 minutes at 2 Hz and include structured visual stimuli such as grating in four directions, bright and dark flashes, and bilateral looming. Stimulus labels are derived from the visual stimulus conditions and used as targets for stimulus classification, whereas behavioral labels are derived from fictive swimming signals recorded through extracellular motor nerve measurements and used as targets for behavior decoding (Appendix~\ref{app:decoding_targets}, Figure~\ref{fig:decodingtarget}).

To construct multimodal biological features, we use each neuron's transcriptomic expression profile, three-dimensional position, anatomical region, and electrophysiological features motivated by met-type profiling~\cite{gouwens2020integrated}. Because jointly acquired gene expression and neural activity data are difficult to obtain in practice, we additionally evaluate our framework using independently acquired, public gene expression data from the mapzebrain atlas. This atlas provides hybridization chain reaction (HCR) in situ RNA profiles~\cite{kunst2019mapzebrain, shainer2023zebrafish_atlas}, which we assign to each neuron based on its atlas coordinates. Feature construction details are provided in Appendix~\ref{app:identity_construction}.

\paragraph{Allen Visual Coding Neuropixels.} 
The Allen Visual Coding Neuropixels dataset~\cite{siegle2021survey} provides large-scale extracellular electrophysiology recordings from the mouse visual cortex. We use both available datasets: the Brain Observatory dataset with 26 sessions and the Functional Connectivity dataset with 25 sessions, comprising over 80,000 units in total. Each session uses six Neuropixels probes to record spiking activity, with unit locations registered to the Allen Common Coordinate Framework~\cite{wang2020allenccf}. Recordings include high-frequency spiking activity over varied durations and cover structured visual stimuli such as drifting gratings, static images, and natural scenes, grouped into five functional classes per dataset. These stimulus categories are used as targets for stimulus classification, and behavioral data consist of running speed, which serves as the target for behavior decoding   (Appendix~\ref{app:decoding_targets}, Figure~\ref{fig:decodingtarget}). This dataset allows us to test whether our framework improves held-out transfer performance even when molecular profiles are unavailable, utilizing only the remaining neuronal identity features described in Appendix \ref{app:identity_construction}.

\subsection{Baselines}
First, we evaluate neural decoding performance by comparing our method to three decoding baselines: POYO~\cite{azabou2023poyo}, NEDS~\cite{zhang2025neds}, and DeepSets~\cite{zaheer2017deep}. POYO uses cross-attention over tokenized neural events, relying on learned session and unit embeddings to aggregate activity across recordings. NEDS uses a transformer with a multi-task masking strategy over neural activity and behavior. Both methods operate across sessions without requiring neuron-wise correspondence, but rely on fine-tuning for transfer to new sessions. DeepSets serves as the permutation-invariant baseline, aggregating neuron-wise activity representations by summation without identity-guided routing. Full details of the baseline adaptation and identity feature integration are provided in Appendix~\ref{app:decoding_baseline_adaptation}. To isolate the effect of the group-routing architecture from access to multimodal biological features, each baseline is evaluated in two variants. The first is its original form, and the second is an augmented form that receives the same per-neuron identity features as our method, with all other training details unchanged. Concretely, neuronal identity features are incorporated into neuron-level embeddings (Appendix~\ref{app:decoding_identity_effect}).

Second, we evaluate neural manifold analysis by comparing the manifolds learned from our method with those produced by CEBRA~\cite{schneider2023cebra}, MARBLE~\cite{gosztolai2025marble}, LFADS~\cite{pandarinath2018lfads}, and Joint PCA.
CEBRA learns low-dimensional embeddings of neural population activity through contrastive learning. MARBLE models local flow fields over the manifold structure using geometric deep learning. LFADS infers single-trial population dynamics through a sequential variational autoencoder. Joint PCA provides the linear baseline, projecting jointly concatenated population activity from source animals into a shared low-dimensional space. Full details of this adaptation are provided in Appendix \ref{app:manifold_baselines}.

\subsection{Evaluation metrics and protocol}
To evaluate decoding performance, we report the Pearson correlation coefficient (PCC) for behavior regression and accuracy scores for stimulus classification, both evaluated on held-out data. 

To assess the neural manifold analysis, we quantify how the latent representation preserves the temporal geometry of population dynamics on held-out data. We compute three metrics on the latent trajectory $\{z_t\}_{t=1}^{T}$ that evaluate global temporal organization and local structure (Appendix~\ref{app:geometric_metrics}).

\begin{itemize}
    \item \textbf{Temporal Distance Correlation} $(\uparrow)$ measures the Spearman correlation between pairwise latent distances $\|z_t - z_s\|_2$ and temporal distances $|t-s|$. Higher values indicate that states nearby in time remain close in latent space while temporally distant states separate.
    \item \textbf{Normalized Latent-Neighbor Time Gap} $(\downarrow)$ measures the temporal distance to each point's nearest neighbor in latent space, normalized by recording length and averaged over time. Lower values indicate higher temporal consistency in  local latent neighborhoods.
    \item \textbf{Local Trajectory Linearity} $(\uparrow)$ measures the fraction of variance explained by the first principal component within short sliding windows of the latent trajectory~\cite{machens2010eff}. Higher values indicate smoother local dynamics.
\end{itemize}

\paragraph{Held-out transfer protocol.}
\label{par:heldout-protocol}
We evaluate transfer at the level of complete held-out recordings. All methods are trained exclusively on source animals or sessions and then applied directly to the held-out recordings. During inference, learned parameters are strictly frozen. At evaluation time, the entire held-out recording is analyzed, leveraging all available data to construct the per-neuron identity features prior to inference (see Appendix~\ref{app:identity_construction} for construction details). All behavioral and stimulus labels from the held-out recordings are used only to compute the reported evaluation metrics. To ensure fair comparison, this exact transfer protocol is maintained across all baseline evaluations.

For the WARP dataset, we train on two fish and evaluate on the remaining fish. We average these results across all three held-out combinations and five random seeds. For the Allen mouse visual cortex electrophysiology datasets, we train on a fixed set of 20 sessions and evaluate independently on the remaining held-out sessions (six for the Brain Observatory and five for the Functional Connectivity datasets), averaging across three random seeds. All identity feature ablations follow the same held-out protocol, differing only in the specific subsets of identity features provided to the model.

\section{Results}
\label{sec:results}

\subsection{Behavior and Stimulus Decoding Under Direct Held-Out Transfer}
\label{sec:decoding}

Table~\ref{tab:decoding} reports behavior and stimulus decoding performance under held-out transfer to unseen animals or sessions without retraining across all three datasets. Our framework achieves the best performance across all datasets and tasks, outperforming prior methods in all cases. On the WARP dataset, our method achieves a PCC value of 0.793. On the Allen Brain Observatory and Functional Connectivity datasets, our method yields PCC values of 0.734 and 0.694, respectively, marking the highest performance in each setting.

Next, we investigate whether the performance gains stem simply from the inclusion of neuronal features, or from how our framework explicitly utilizes them. Accordingly, we provide POYO, NEDS, and DeepSets with the same multimodal biological features, leaving their original implementation details unchanged (Appendix~\ref{app:decoding_baseline_adaptation}). Across datasets, these identity-augmented baselines show consistent improvements over their original counterparts. Behavior decoding on the WARP dataset increases from 0.304 to 0.489 for POYO, from 0.248 to 0.574 for NEDS, and from 0.160 to 0.466 for DeepSets. Similar trends are observed for stimulus decoding and for the Allen datasets, where identity features are based mainly on position and anatomical region (Appendix~\ref{app:decoding_identity_effect}). However, their performance remains below that of our method across all datasets. This pattern demonstrates that our framework's superior performance cannot be fully explained by just the availability of these biological features. Instead, the key advantage lies in how they are incorporated through identity-guided organization.

We further assess our framework's robustness to the choice and availability of identity modalities by evaluating ablations where we remove one modality at a time, reported in Appendix~\ref{app:ablations}. Across both datasets, all ablated variants remain competitive and continue to outperform baseline methods, showing that the performance gains are not attributable to any single modality.
The model without transcriptomic descriptors achieves a behavior PCC of 0.727 and a stimulus balanced accuracy of 0.711, compared with 0.793 and 0.835 for the full model. This suggests that the framework remains competitive even when costly jointly acquired transcriptomic measurements are unavailable.

\begin{table}[t!]
\centering
\caption{Behavior and stimulus decoding under direct held-out transfer across three datasets
(mean $\pm$ std). Bold indicates the best result
per column within each dataset. * denotes models adapted to calcium signals.}
\label{tab:decoding}
\smallskip
\resizebox{0.95\linewidth}{!}{%
\begin{tabular}{llccc}
\toprule
Dataset & Method
    & Behavior (PCC) $\uparrow$
    & Stimulus (Balanced Accuracy) $\uparrow$
    & Stimulus (Accuracy) $\uparrow$ \\
\midrule
\multirow{4}{*}{WARP}
 & \textbf{Ours}
    & $\mathbf{0.793 \pm 0.015}$
    & $\mathbf{0.835 \pm 0.029}$
    & $\mathbf{0.839 \pm 0.029}$ \\
 & {POYO*~\cite{azabou2023poyo}}
    & {${0.304 \pm 0.120}$} & {${0.153 \pm 0.032}$} & {${0.156 \pm 0.034}$} \\
 & {NEDS*~\cite{zhang2025neds}}
    & {${0.248 \pm 0.207}$} & {${0.195 \pm 0.034}$} & $0.196 \pm 0.033$ \\
 & DeepSets~\cite{zaheer2017deep}
    & $0.160 \pm 0.081$
    & $0.161 \pm 0.011$ & $0.160 \pm 0.011$\\
\midrule
\multirow{4}{*}{\shortstack[l]{Allen Brain\\Observatory}}
 & \textbf{Ours}
    & $\mathbf{0.734 \pm 0.008}$
    & $\mathbf{0.893 \pm 0.012}$ 
    & $\mathbf{0.899 \pm 0.010}$ \\
 & {POYO~\cite{azabou2023poyo}}
    & {${0.526 \pm 0.043}$} & {${0.220 \pm 0.022}$} & $0.258 \pm 0.114$ \\
  & {NEDS~\cite{zhang2025neds}}
    & {${0.500 \pm 0.042}$} & {${0.424 \pm 0.043}$} & $0.432 \pm 0.078$ \\
 & DeepSets~\cite{zaheer2017deep}
    & $0.615 \pm 0.022$
    & $0.684 \pm 0.027$ & $0.612 \pm 0.034$ \\
\midrule
\multirow{4}{*}{\shortstack[l]{Allen Functional\\Connectivity}}
 & \textbf{Ours}
    & $\mathbf{0.694 \pm 0.058}$
    & $\mathbf{0.896 \pm 0.017}$ & $\mathbf{0.932 \pm 0.007}$ \\
 & {POYO~\cite{azabou2023poyo}}
    & {${0.554 \pm 0.080}$} & {${0.211 \pm 0.008}$} & $0.288 \pm 0.117$ \\
 & {NEDS~\cite{zhang2025neds}}
    & {${0.454 \pm 0.036}$} & {${0.398 \pm 0.020}$} & $0.408 \pm 0.091$ \\
 & DeepSets~\cite{zaheer2017deep}
    & $0.588 \pm 0.096$
    & $0.730 \pm 0.044$ & $0.663 \pm 0.095$ \\
\bottomrule
\end{tabular}
}
\end{table}

\subsection{Neural Manifold Geometry Under Direct Held-Out Transfer}
\label{sec:manifold}

Next, we evaluate the performance of our framework in neural manifold analysis, specifically testing its ability to preserve temporal and local geometric structure on unseen recordings without retraining, a capability that has remained largely beyond the reach of existing methods. This capability is crucial for rigorous comparative analysis, allowing researchers to investigate how neural population dynamics vary or remain invariant across different subjects under identical computational constraints.

Table~\ref{tab:manifold} reports temporal and local geometric metrics on held-out animals under the same held-out transfer protocol. On the WARP dataset, our framework achieves a temporal distance correlation of $0.633 \pm 0.078$, indicating that neural states nearby in time remain close in latent space. It also obtains the lowest normalized nearest-neighbor temporal gap, $0.145 \pm 0.009$, and the highest local trajectory linearity, $0.911 \pm 0.006$.
For the manifold baselines, we additionally test \(1,000\) random neuron orderings at inference time and find that our method outperforms the best sampled baseline results in most cases (Appendix~\ref{app:baseline_ordering}).
A consistent pattern is observed on the Allen mouse visual cortex datasets, where our framework again outperforms all baselines across the evaluated metrics. Figure~\ref{figure_manifold_time} shows that held-out trajectories remain temporally coherent and avoid the fragmentation or collapse observed in several baseline embeddings. Trajectories colored by behavior and stimulus labels are shown in Figs.~\ref{figure_manifold_behavior} and~\ref{figure_manifold_stimulus}, respectively.

Together, these results indicate that our framework produces transferred latent trajectories with more consistent temporal neighborhoods and smoother local structure in held-out recordings under the evaluated metrics. This advantage is consistent with a key difference in how the methods handle unseen animals. Baseline methods learn representations tied to the specific set of neurons recorded in the source animals and lack a way to align these representations when the neuron population changes across individuals. Our framework instead grounds each neuron in a common identity coordinate, so the model induces a shared identity-defined group coordinate that supports transfer to held-out animals and sessions.

\begin{table}[t]
\centering
\caption{Geometric structure of latent trajectories evaluated on held-out animals or sessions across three datasets (mean $\pm$ std). Bold indicates the best result per column.}
\label{tab:manifold}
\resizebox{0.92\linewidth}{!}{%
\begin{tabular}{llccc}
\toprule
Dataset & Method
    & \shortstack[c]{Temporal Distance\\Correlation $\uparrow$}
    & \shortstack[c]{Normalized Latent-Neighbor\\Time Gap $\downarrow$}
    & \shortstack[c]{Local Trajectory\\Linearity $\uparrow$} \\
\midrule
\multirow{5}{*}{WARP} & \textbf{Ours} & $\mathbf{0.633 \pm 0.078}$ & $\mathbf{0.145 \pm 0.009}$ & $\mathbf{0.911 \pm 0.006}$ \\
 & CEBRA~\cite{schneider2023cebra} & $0.086 \pm 0.010$ & $0.328 \pm 0.003$ & $0.684 \pm 0.010$ \\
 & MARBLE~\cite{gosztolai2025marble} & $0.045 \pm 0.017$ & $0.332 \pm 0.002$ & $0.512 \pm 0.006$ \\
 & LFADS~\cite{pandarinath2018lfads} & $0.141 \pm 0.020$ & $0.297 \pm 0.012$ & $0.666 \pm 0.015$ \\
 & Joint PCA & $0.151 \pm 0.045$ & $0.312 \pm 0.007$ & $0.887 \pm 0.032$ \\
\midrule
\multirow{5}{*}{\shortstack[l]{Allen Brain\\Observatory}} & \textbf{Ours} & $\mathbf{0.355 \pm 0.033}$ & $\mathbf{0.063 \pm 0.021}$ & $\mathbf{0.896 \pm 0.005}$ \\
 & CEBRA~\cite{schneider2023cebra} & $0.029 \pm 0.004$ & $0.298 \pm 0.005$ & $0.497 \pm 0.006$ \\
 & MARBLE~\cite{gosztolai2025marble} & $0.049 \pm 0.004$ & $0.190 \pm 0.009$ & $0.654 \pm 0.004$ \\
 & LFADS~\cite{pandarinath2018lfads} & $0.179 \pm 0.011$ & $0.134 \pm 0.009$ & $0.545 \pm 0.004$ \\
 & Joint PCA & $0.071 \pm 0.019$ & $0.244 \pm 0.017$ & $0.699 \pm 0.029$ \\
\midrule
\multirow{5}{*}{\shortstack[l]{Allen Functional\\Connectivity}} & \textbf{Ours} & $\mathbf{0.317 \pm 0.036}$ & $\mathbf{0.059 \pm 0.019}$ & $\mathbf{0.893 \pm 0.013}$ \\
 & CEBRA~\cite{schneider2023cebra} & $0.025 \pm 0.006$ & $0.298 \pm 0.015$ & $0.494 \pm 0.007$ \\
 & MARBLE~\cite{gosztolai2025marble} & $0.047 \pm 0.005$ & $0.196 \pm 0.014$ & $0.663 \pm 0.021$ \\
 & LFADS~\cite{pandarinath2018lfads} & $0.150 \pm 0.012$ & $0.126 \pm 0.019$ & $0.507 \pm 0.009$ \\
 & Joint PCA & $0.052 \pm 0.017$ & $0.226 \pm 0.041$ & $0.682 \pm 0.030$ \\

\bottomrule
\end{tabular}
}
\end{table}

\begin{figure}[t]
    \centering
    \centerline{\includegraphics[width=0.84\columnwidth]{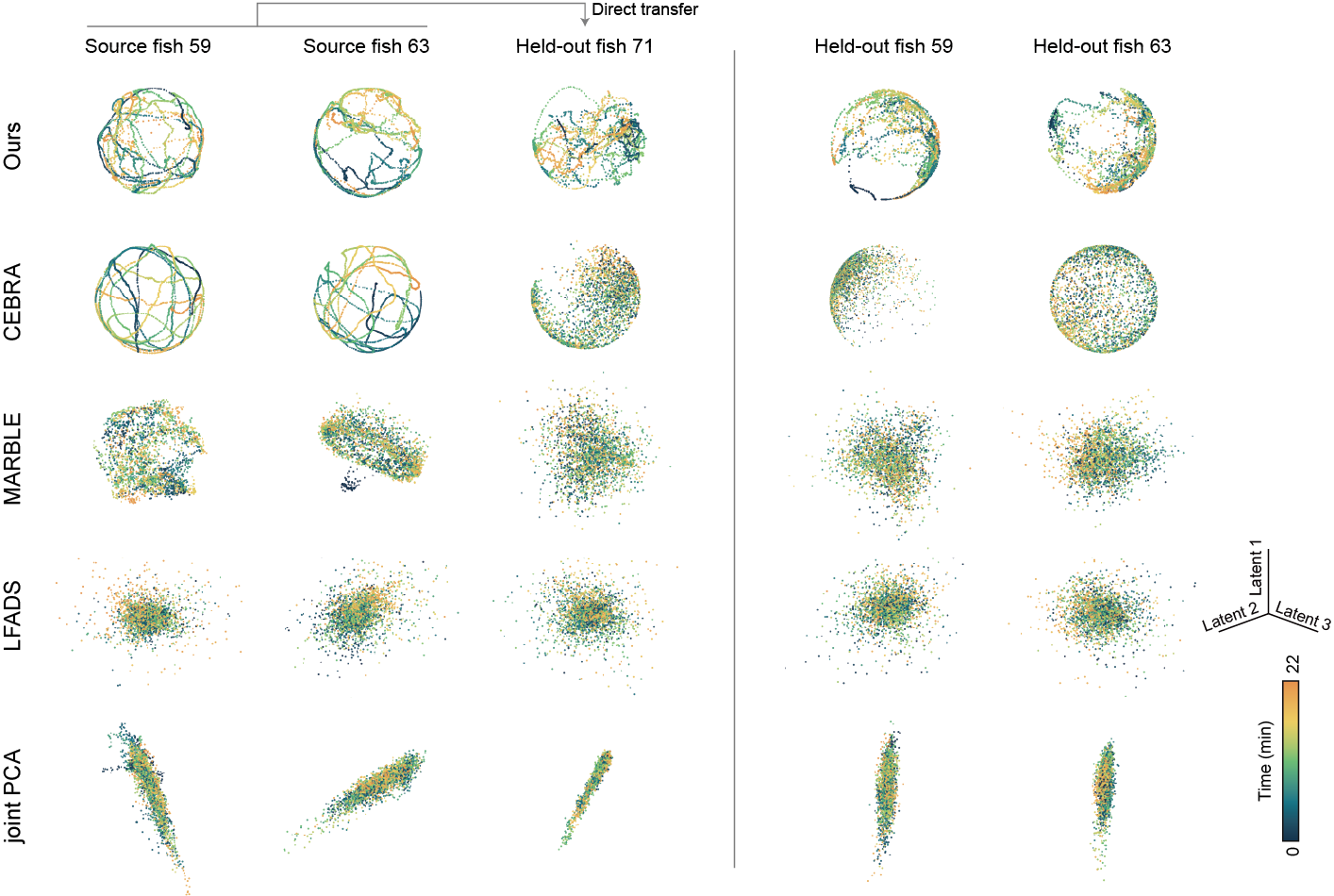}}
    \caption{Latent trajectories for two source animals (Fish 59 and Fish 63) and held-out animals across all methods. Points are colored by time within each recording. 
    Our method preserves structured temporal trajectories in held-out animals, whereas baseline methods yield fragmented or collapsed trajectories under direct held-out transfer.
}
    \label{figure_manifold_time}
\end{figure}

\section{Related Work}

\paragraph{Neural population dynamics analysis methods.}
Recent work in neural population dynamics has increasingly explored the use of deep neural networks~\citep{shenoy2013dynamical,shenoy2021brainwide}. Methods such as POYO, POCO, and NEDS learn representations that aggregate neural activity across recordings through tokenization or latent mapping, enabling cross-session and cross-task decoding without requiring neuron-wise correspondence \citep{azabou2023poyo,zhang2025neds,azabou2025multisession,duan2025poco}. Across these methods, neurons are represented through learned, session-specific representations derived from observed dynamics and implicit neuron indexing, without explicit grounding in biological identity. While some recent work has begun to explore learning neuron identity from population context \citep{arora2025know}, cross-session generalization still primarily relies on statistical regularities in observed dynamics rather than on molecular or spatial information that makes neurons comparable across individuals.

\paragraph{Neural manifold methods.}
A central observation in neuroscience is that the high-dimensional activity of neural populations often lies on low-dimensional manifolds whose geometry reflects the structure of the underlying dynamics \citep{schneider2023cebra,gosztolai2025marble,pandarinath2018lfads}. Methods for characterizing these manifolds include sequential variational autoencoders, contrastive embedding approaches that align neural activity with behavioral or stimulus variables \citep{schneider2023cebra}, and geometric deep learning methods that analyze local flow fields over manifold structure \citep{gosztolai2025marble}. These approaches have become standard tools for relating population-level dynamics to cognition, perception, and movement \citep{shenoy2013dynamical,shenoy2021brainwide}.

\paragraph{Multimodal neuron profiling and cell-type atlases.}
Recent work in systems neuroscience shows that neurons can be characterized along multiple biological axes. Multimodal profiling approaches, including Patch-seq and spatial transcriptomics, enable large-scale measurement of molecular and spatial properties \citep{gouwens2020integrated,biccn2021motorcortex,cadwell2024patchseq,fuzik2016integration,shao2024patchseqreview,yagishita2025integrating}. Large-scale atlas efforts have produced spatially resolved gene expression maps and anatomical reference frameworks at cellular resolution across whole-brain volumes \citep{zeisel2018molecular,tasic2018shared,yao2021taxonomy,kunst2019mapzebrain,shainer2023zebrafish_atlas,han2025mousebrainst}. More recent datasets jointly profile or align neural activity with transcriptomic identity within the same brain \citep{biccn2021motorcortex,marquezlegorreta2026warp}. Across species, multimodal measurements show that molecular and spatial identity are associated with functional circuit organization \citep{arkhipov2025multimodal_circuits,liu2025connectivity}. This growing availability of per-neuron identity information provides a natural foundation for incorporating biological structure into methods for analyzing population dynamics.

\section{Discussion}
We propose multimodal biological features as an organizational basis for analyzing neural population dynamics. Our method incorporates these features to route neurons into latent groups and analyze dynamics at the group level. This design preserves invariance to neuron order while enabling neurons with similar biological context to contribute through a shared functional organization. As a result, populations from new animals or sessions can be mapped into a common identity-defined structure without requiring neuron-wise correspondence or retraining.

Across zebrafish and mouse visual cortex datasets, our method achieves strong behavior and stimulus decoding while producing geometrically consistent neural manifolds under direct held-out transfer. Baseline methods also improve when provided with the same neuronal identity features, yet our routing-based framework consistently outperforms them. This suggests that effective generalization depends not only on the availability of neuronal identity information, but also on how it is used to impose a shared organization on population dynamics.

Our method also remains effective when identity features are not jointly recorded with neural activity, performing comparably when using independently acquired atlas data. This expands applicability to existing datasets where neural activity is available but neuron-level biological feature measurements are not collected in the same animals.

The performance of the method is shaped by the availability and quality of identity features. Joint measurements of neural activity and transcriptomic identity remain limited \citep{biccn2021motorcortex,marquezlegorreta2026warp,cadwell2024patchseq,fuzik2016integration,yagishita2025integrating}, and using independently acquired atlas data introduces additional challenges in registration \citep{clifton2023stalign,cho2023robust,zhang2024brainalign}. Variability across individuals and imperfections in alignment can introduce noise into identity features, affecting group assignments. Improvements in multimodal registration and joint data acquisition are therefore likely to further strengthen identity-based population analysis.

More broadly, these results highlight shared biological structure as an underexplored signal for neural population analysis. Neuronal identity provides a complementary axis for organizing population representations across individuals, suggesting a path toward methods that can integrate data across subjects without requiring exact neuron-wise correspondence.

\textbf{Acknowledgements}. This research was supported by the National Research Foundation of Korea (NRF) (RS-2026-25475246, RS-2026-25504129, RS-2026-25475217, RS-2023-00209473, RS-2023-00264409, RS-2026-25498280, RS-2022-NR068758), Korea Basic Science Institute (National Research Facilities and Equipment Center) grant funded by the Ministry of Science and ICT (RS-2024-00401676), Bio \& Medical Technology Development Program through the NRF funded by the Ministry of Science and ICT (RS-2021-NR056586), the BK21 plus program through the NRF funded by the Ministry of Education of Korea, a Research Grant from HFSP (RGP003/2024) with the award DOI (https://doi.org/10.52044/HFSP.RGP0032024.pc.gr.194148), the KAIST Quantum+X Convergence R\&D Project, the KAIST C2 (Creative \& Challenging) Project, the KAIST Jang Young Sil Fellow Program, and KAIST National Research and Development Project (N10260124), and KAIST AI-Science Hub Project (N10260124).

\bibliographystyle{unsrtnat}
\bibliography{reference}

\newpage
\appendix

\renewcommand{\thefigure}{A\arabic{figure}}
\renewcommand{\thetable}{A\arabic{table}}
\setcounter{figure}{0}
\setcounter{table}{0}

\section{Supplementary Results and Implementation Details}

\subsection{Dataset overview}

\paragraph{Multimodal information across recordings.}
\Cref{fig:conservewarp,fig:conserveallen} show representative WARP fish and Allen
visual cortex sessions. Although the exact recorded neurons or units differ
across animals and sessions, both datasets exhibit broadly comparable spatial
and regional organization. Fish in the WARP dataset show similar whole-brain distributions
of transcriptomic markers and brain-region labels, while Allen sessions cover
comparable visual cortical and associated areas. This shared structure motivates
our held-out transfer evaluation, where models are applied to held-out recordings
without neuron-wise correspondence.

\begin{figure}[h]
    \centering
    \centerline{\includegraphics[width=0.8\columnwidth]{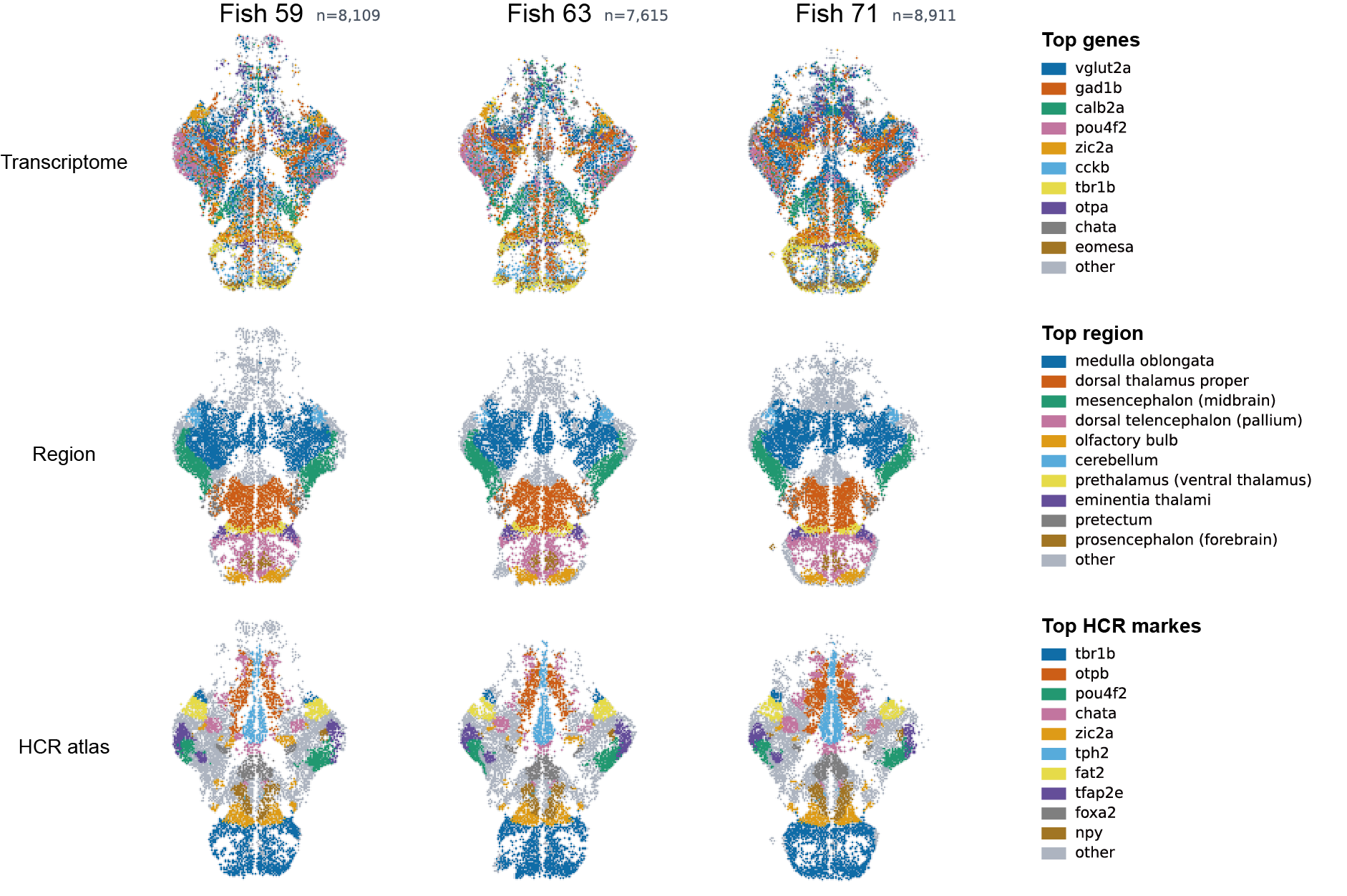}}
    \caption{Data from three fish in the WARP dataset. Neurons are shown in the registered
zebrafish brain coordinate space and colored by transcriptomic marker identity
(top), anatomical region label (middle), and selected marker-gene categories
(bottom).
}
    \label{fig:conservewarp}
\end{figure}

\begin{figure}[h]
    \centering
    \centerline{\includegraphics[width=0.8\columnwidth]{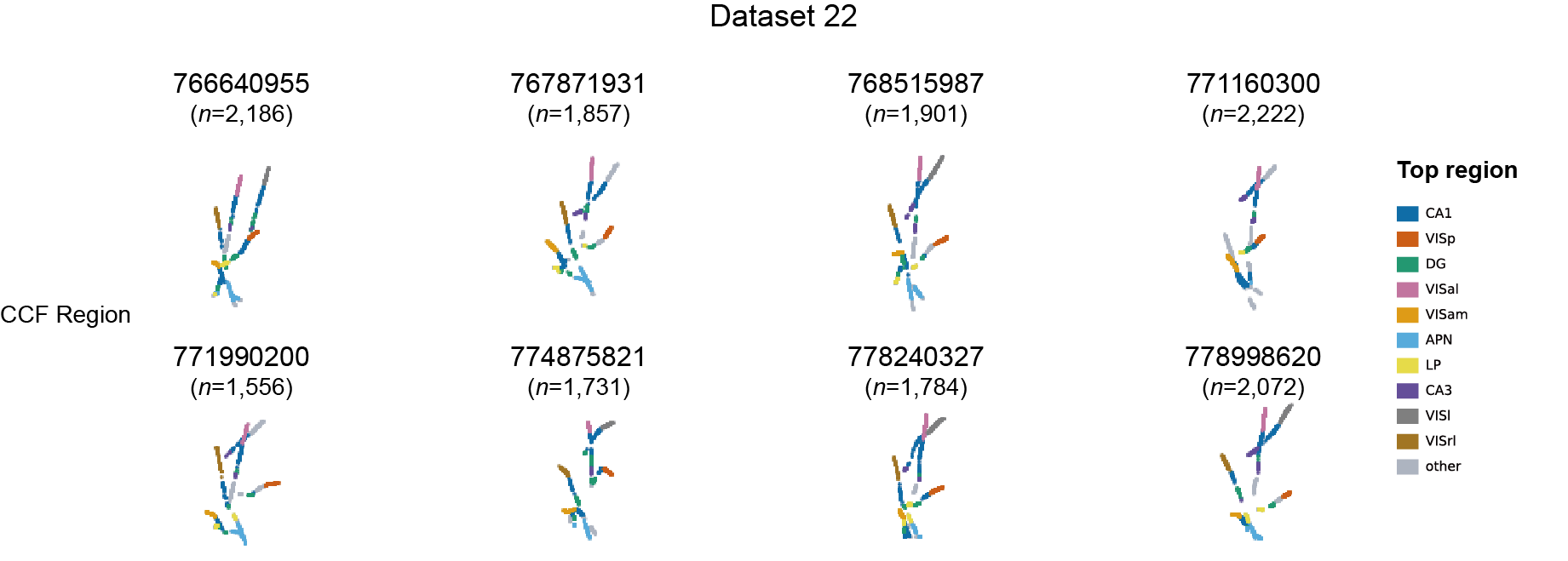}}
    \caption{Representative Allen visual cortex sessions shown in a shared
anatomical coordinate space and colored by visual area label.
}
    \label{fig:conserveallen}
\end{figure}

\paragraph{Decoding Targets.}
\label{app:decoding_targets}
Figure~\ref{fig:decodingtarget} summarizes the supervised decoding
targets used in each dataset. Behavior decoding is formulated as scalar
regression and evaluated with Pearson correlation coefficient. Stimulus decoding
is formulated as multi-class classification and evaluated with accuracy. In WARP, the behavioral target is swim power, and the stimulus target
is the type of visual stimuli. In the Allen datasets, the
behavioral target is running speed, and the stimulus target is the type of visual stimuli. These targets are used consistently for the proposed
method and all supervised decoding baselines.

\begin{figure}[h]
    \centering
    \centerline{\includegraphics[width=0.9\columnwidth]{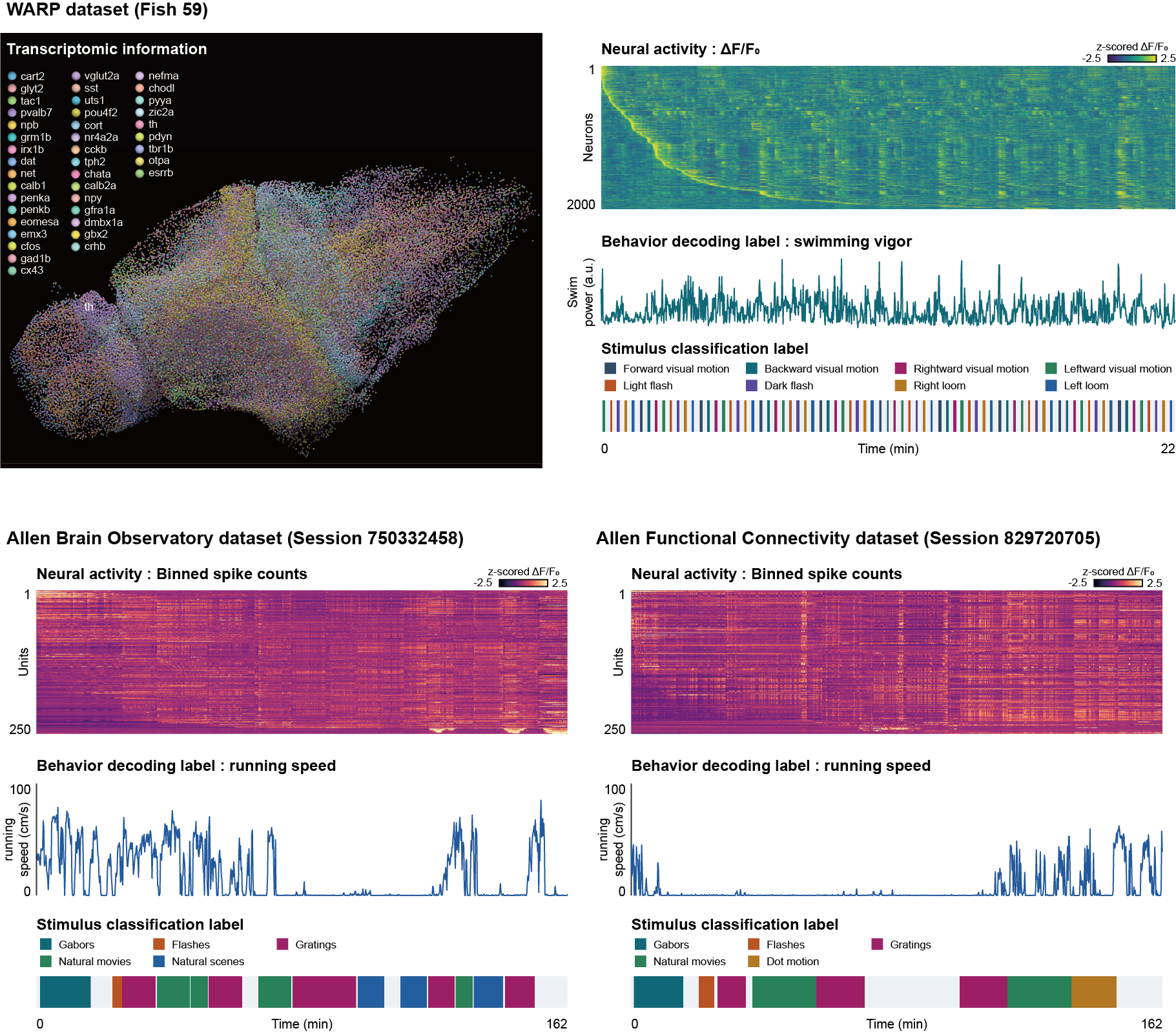}}
    \caption{Decoding targets used in WARP and Allen experiments. For WARP, we show a representative fish with neural activity from selected neurons, the swim-power behavioral  target, and the visual-stimulus classification target. For Allen, we show representative sessions from Allen Brain Observatory dataset and Allen Functional Connectivity dataset with binned spike counts from selected units, running-speed behavioral targets, and grouped stimulus-family labels. Behavior decoding is evaluated with Pearson correlation coefficient, and stimulus decoding is evaluated with  accuracy.
}
    \label{fig:decodingtarget}
\end{figure}

\subsection{Proof of Permutation Invariance}
\label{app:perm-inv}

In Section~\ref{sec:framework}, the neural population is represented as the
unordered collection
\[
    \mathcal{S}
    =
    \{(\mathbf{f}_i,\mathbf{x}_i)\}_{i=1}^{N}.
\]
For the proof, we make the arbitrary enumeration of this collection explicit and
write
\[
    \mathbf{S}
    =
    \big((\mathbf{f}_1,\mathbf{x}_1),\ldots,(\mathbf{f}_N,\mathbf{x}_N)\big).
\]
For any permutation \(\pi \in \mathfrak{S}_N\), define the reindexed input
\[
    \mathbf{S}^{\pi}
    =
    \big((\mathbf{f}_{\pi(1)},\mathbf{x}_{\pi(1)}),\ldots,
    (\mathbf{f}_{\pi(N)},\mathbf{x}_{\pi(N)})\big).
\]

\begin{proposition}[Permutation invariance]
Let \(\mathcal{F}\) be the population encoder defined by the identity-guided
routing in Eq.~\eqref{eq:routing}, the shared dynamics encoder in
Eq.~\eqref{eq:dynamics-enc}, the group-level aggregation in
Eq.~\eqref{eq:aggregation}, and the group interaction module \(\Gamma\). If the
identity encoder, dynamics encoder, and routing map are shared across neurons,
then for any permutation \(\pi \in \mathfrak{S}_N\),
\[
    \mathcal{F}(\mathbf{S}^{\pi})
    =
    \mathcal{F}(\mathbf{S}).
\]
\end{proposition}

\begin{proof}
Let the input neuron ordering be permuted by \(\pi\). Since the identity encoder,
dynamics encoder, and routing map are shared across neurons, the corresponding
neuron-level quantities are permuted in the same way:
\[
    \mathbf{s}^{\pi}_i = \mathbf{s}_{\pi(i)},
    \qquad
    \mathbf{v}^{\pi}_i = \mathbf{v}_{\pi(i)},
    \qquad
    a^{\pi}_{ik} = a_{\pi(i)k}.
\]
For each latent group \(k\), the effective group mass under the permuted input is
\[
    m_k^{\pi}
    =
    \sum_{i=1}^{N} a^{\pi}_{ik}
    =
    \sum_{i=1}^{N} a_{\pi(i)k}
    =
    \sum_{j=1}^{N} a_{jk}
    =
    m_k,
\]
where the third equality follows by reindexing \(j=\pi(i)\). Similarly,
\[
    \sum_{i=1}^{N} a^{\pi}_{ik}\mathbf{v}^{\pi}_i
    =
    \sum_{i=1}^{N} a_{\pi(i)k}\mathbf{v}_{\pi(i)}
    =
    \sum_{j=1}^{N} a_{jk}\mathbf{v}_{j}.
\]
Therefore,
\[
    \mathbf{c}^{\pi}_k
    =
    \frac{1}{\sqrt{\max(m_k,\epsilon)}}
    \sum_{i=1}^{N} a^{\pi}_{ik}\mathbf{v}^{\pi}_i
    =
    \frac{1}{\sqrt{\max(m_k,\epsilon)}}
    \sum_{j=1}^{N} a_{jk}\mathbf{v}_{j}
    =
    \mathbf{c}_k.
\]
Thus \(\mathbf{C}^{\pi}=\mathbf{C}\). Since the final representation is computed
as a deterministic function of the group matrix,
\[
    \mathbf{z}
    =
    \Gamma(\mathbf{C}),
\]
we have
\[
    \mathbf{z}^{\pi}
    =
    \Gamma(\mathbf{C}^{\pi})
    =
    \Gamma(\mathbf{C})
    =
    \mathbf{z}.
\]
Hence, \(\mathcal{F}(\mathbf{S}^{\pi})=\mathcal{F}(\mathbf{S})\), and
\(\mathcal{F}\) is permutation-invariant.
\end{proof}

\subsection{Neuronal Identity Feature Construction}
\label{app:identity_construction}
We designed the identity encoder framework to serve as a versatile interface for any available neuron-level features, allowing the set of input modalities to vary across datasets and experimental conditions. Table~\ref{tab:identity_features} summarizes all feature groups that we used and their availability across datasets. For modality-ablation experiments, we use the same construction procedure and vary only which groups are provided.
For the WARP dataset, spatial coordinates are registered to the mapzebrain reference spaces~\cite{kunst2019mapzebrain}, and jointly acquired transcriptomic marker features are available alongside atlas-derived HCR expression features. For the Allen dataset, joint transcriptomic measurements are not available, and spatial coordinates are registered to the Allen Common Coordinate Framework~\cite{wang2020allenccf}.
\paragraph{mapzebrain HCR gene expression atlas features.}
\label{app:hcr-atlas}
To test whether independently acquired, public atlas information could provide useful spatial gene-expression context for WARP neurons, we derived atlas-based features from the mapzebrain single-cell-resolution gene-expression atlas~\cite{kunst2019mapzebrain, shainer2023zebrafish_atlas}, which uses multiplexed fluorescent in situ RNA hybridization chain reaction (HCR) to generate spatial expression maps across the larval zebrafish brain. The atlas contains 236 HCR marker volumes in the mapzebrain reference space. For each WARP neuron, marker intensity was sampled at its registered position using a local Gaussian-weighted average (std.\ 8 voxels, truncated at 2 std., with boundary normalization), yielding a 236-dimensional feature vector used as independently derived spatial gene-expression covariates.
\begin{table}[h]
\small
\centering
\caption{Identity feature groups for each dataset.}
\label{tab:identity_features}
\smallskip
\begin{tabular}{p{2.8cm}p{7cm}cc}
\toprule
Feature Group & Description & WARP & Allen \\
\midrule
Transcriptomics
    & 41-dimensional gene-expression marker vector
    & \checkmark
    & \\[4pt]
Spatial position
    & 3D registered coordinate in reference space
    & \checkmark
    & \checkmark \\[4pt]
Anatomical region
    & Multihot vector over brain-region labels
    & \checkmark
    & \checkmark \\[4pt]
HCR expression
    & 236-dimensional atlas-derived gene-expression marker vector as documented in Appendix~\ref{app:hcr-atlas}
    & \checkmark
    & \\[4pt]
Electrophysiological features
    & Features from observed dynamics including mean, standard deviation, extrema, quartiles, positive fraction, mean squared value, and the standard deviation of first differences
    & \checkmark
    & \checkmark \\
\bottomrule
\end{tabular}
\end{table}

\subsection{Geometric Trajectory Metrics}
\label{app:geometric_metrics}

We provide the details of the geometric metrics used to evaluate
whether transferred latent trajectories preserve temporal organization and local
trajectory structure on held-out recordings. Let
\(\{\mathbf{z}_t\}_{t=1}^{T}\), with \(\mathbf{z}_t \in \mathbb{R}^{d}\), denote
the latent trajectory for a held-out recording. Each metric is computed
separately for each held-out trajectory and then averaged across held-out
animals or sessions.

\paragraph{Temporal distance correlation.}
This metric measures whether distances in latent space increase monotonically
with temporal separation.
We sample \(20{,}000\) frame pairs by drawing a frame
index \(i\) and an offset
\(\Delta \in \{1,\ldots,\min(100,T-1)\}\), retaining valid pairs
\((i,j)\) with \(j=i+\Delta \leq T\). We then compute the Spearman correlation
between the latent distances \(\|\mathbf{z}_i-\mathbf{z}_j\|_2\) and the
temporal gaps \(|i-j|\). Higher values indicate that temporal separation is more
consistently reflected in the latent geometry.

\paragraph{Normalized latent-neighbor time gap.}
This metric measures whether local neighborhoods in latent space correspond to
nearby time points. For each frame \(t\), we find its \(k=10\) nearest neighbors
in latent space, excluding the point itself, and compute the average temporal
gap to these neighbors:
\[
    \frac{1}{Tk}
    \sum_{t=1}^{T}
    \sum_{j \in \mathcal{N}_{10}(t)}
    |t-j|.
\]
We then divide this mean gap by \(T\) to report the metric on a normalized \([0,1]\) scale. Lower values indicate that nearby latent
states tend to correspond to nearby time points.

\paragraph{Local trajectory linearity.}
This metric quantifies the local one-dimensionality of short latent trajectory segments. We use sliding windows
of length \(25\) frames with stride \(5\). For each window \(W\), we center the
latent states in the window, compute their covariance matrix, and let
\(\lambda_1 \geq \lambda_2 \geq \cdots\) denote its positive eigenvalues. The
window-level linearity score is
\[
    \ell(W)
    =
    \frac{\lambda_1}{\sum_r \lambda_r}.
\]
The reported local trajectory linearity is obtained by averaging \(\ell(W)\)
over all sliding windows. Larger values indicate that local trajectory variance
is concentrated along a single dominant principal direction.

\subsection{Implementation Details for the Proposed Method}
\label{app:proposed_model_settings}

All experiments use the identity-guided population method described in
Sec.~\ref{sec:framework}.  At each evaluation time point, neuronal activity and
neuronal identity features are provided as input. Routing weights
are computed only from neuronal identity features, and neuron-level activity
representations are aggregated into latent groups using the resulting
identity-guided routing weights.  The group representations are  processed by the group interaction module
$\Gamma$.
Here, $\Gamma$ is implemented as an
MLP mixer over the latent group axis, whose output dimension is set by the
corresponding objective, following the common contrastive manifold learning setup of
CEBRA~\cite{schneider2023cebra}.

The same architecture is used for supervised decoding and manifold learning, with task-specific output dimensions, training objectives, and hyperparameters (Tables A2-A3). For
behavior decoding, $\Gamma$ outputs one scalar. For stimulus decoding, $\Gamma$
outputs a vector whose dimension is the number of stimulus classes. For manifold
learning, $\Gamma$ outputs a three-dimensional latent vector
$\mathbf{z}_t \in \mathbb{R}^{3}$, which is $\ell_2$-normalized before applying
the contrastive objective.

When a recording contains more neurons than the specified neurons-per-subset
value in Tables~\ref{tab:decoding_hparams}--\ref{tab:manifold_hparams}, we apply
the model to uniformly sampled neuron subsets. During training, one subset is
sampled per minibatch. During evaluation, we use a fixed set of sampled subsets and aggregate their outputs: predictions are averaged for behavior regression, class predictions are combined by majority vote for stimulus classification, and
latent states are averaged for manifold analysis. All reported metrics are computed from these aggregated outputs. The same subsampling and evaluation ensembling protocol is used for all decoding baseline methods.

\paragraph{Neural decoding.}

Neural decoding follows the held-out transfer protocol in Section~\ref{par:heldout-protocol}.

For WARP,
each method is trained on source fish and evaluated on the held-out fish. For the
Allen datasets, the model is trained on source sessions and evaluated on
held-out sessions.  For each Allen dataset, we randomly select the source and
held-out sessions once and use the same fixed split for all methods.
Model parameters are fixed during held-out evaluation, and no
held-out behavioral or stimulus labels are used for training, fine-tuning, or
model selection. 

For behavior decoding, the model predicts the
standardized target and is trained with the Huber loss
\[
    \mathcal{L}_{\mathrm{beh}}
    =
    \frac{1}{|\mathcal{B}|}
    \sum_{t \in \mathcal{B}}
    h_{\delta}(\hat{y}_t - y_t),
    \qquad
    h_{\delta}(e)
    =
    \begin{cases}
        \frac{1}{2} e^2, & |e| \leq \delta, \\
        \delta\left(|e| - \frac{1}{2}\delta\right), & |e| > \delta,
    \end{cases}
\]
with $\delta=1.0$. Here, $\mathcal{B}$ denotes the training batch, $\hat{y}_t$
is the predicted standardized behavioral target, $y_t$ is the standardized
ground-truth target, and $e=\hat{y}_t-y_t$ is the prediction error. 
For stimulus decoding, the model is trained with the standard cross-entropy loss
over stimulus classes. We use AdamW~\cite{loshchilov2017decoupled} and early stopping based on the validation score, with exponential moving average
smoothing coefficient $\alpha=0.5$. Hyperparameters are listed in
Table~\ref{tab:decoding_hparams}.

\begin{table}[t]
\centering
\small
\caption{Hyperparameters for supervised decoding with the proposed method.}
\label{tab:decoding_hparams}
\begin{tabular}{ll}
\toprule
Parameter & Value \\
\midrule
Latent groups $K$ & 256 \\
Activity encoder hidden size & 128 \\
Group activity dimension & 384 \\
Identity hidden size & 192 \\
Group mixer hidden size & 64 \\
Routing temperature & 0.35 \\
Routing sparsity & Top-5 soft assignment \\
Batch size & 64 \\
Optimizer & AdamW \\
Learning rate & $2 \times 10^{-4}$ \\
Weight decay & 0 \\
Maximum training epochs & 180 \\
WARP neurons per subset & 3000 train / 3000 eval \\
Allen neurons per subset & 250 train / 250 eval \\
Evaluation ensemble subsets & 21 \\
\bottomrule
\end{tabular}
\end{table}

\paragraph{Manifold learning.}

For manifold learning, the model is trained on source recordings using a time-contrastive objective following the CEBRA implementation
\cite{schneider2023cebra}. For each anchor time $t$, a positive time point \(t+\Delta\) is  chosen within a fixed temporal offset   and negative time points \(\mathcal{N}_t\) are drawn from other time points
in the training batch. Let
\(\mathbf{z}_t\) denote the normalized latent embedding of the
population state at time \(t\), and let \(s(\cdot,\cdot)\) denote cosine
similarity. We use the CEBRA-style generalized InfoNCE objective, in which the
positive pair contributes an attraction term and the log-partition term is
estimated over negative samples:
\[
    \mathcal{L}_{\mathrm{manifold}}
    =
    -s(\mathbf{z}_t,\mathbf{z}_{t+\Delta})/\tau_{nce}
    +
    \log \sum_{u \in \mathcal{N}_t}
    \exp\!\left(s(\mathbf{z}_t,\mathbf{z}_u)/\tau_{nce}\right)
\]
where \(\tau_{nce}\) is the contrastive temperature and \(\mathcal{N}_t\) denotes the
negative set only. This follows the CEBRA convention in which the positive
sample is not included in the log-sum-exp term. 
After training on the source fish, the learned
mapping is applied directly to the held-out fish without retraining or
neuron-wise alignment. Hyperparameters are listed in
Table~\ref{tab:manifold_hparams}.

\begin{table}[t]
\centering
\small
\caption{Hyperparameters for manifold learning with the proposed method.}
\label{tab:manifold_hparams}
\begin{tabular}{ll}
\toprule
Parameter & Value \\
\midrule
Output dimension & 3 \\
Latent groups $K$ & 32 \\
Activity encoder hidden size & 16 \\
Group activity dimension & 32 \\
Identity hidden size & 16 \\
Group mixer hidden size & 32 \\
Routing temperature & 0.35 \\
Routing sparsity & Top-4 soft assignment \\
Contrastive similarity & Cosine \\
Contrastive temperature $\tau_{nce}$ & 1.0 \\
Positive temporal offset $\Delta$ & 5 frames \\
Batch size & 128 \\
Optimizer & AdamW \\
Learning rate & $1 \times 10^{-3}$ \\
Weight decay & 0 \\
Training epochs & 3000 \\
WARP neurons per subset & 3000 train / 3000 eval \\
Allen neurons per subset & 250 train / 250 eval \\
Evaluation ensemble subsets & 11 \\
\bottomrule
\end{tabular}
\end{table}

\subsection{Baseline methods for decoding}
\label{app:decoding_baseline_adaptation}
We evaluated DeepSets~\cite{zaheer2017deep}, POYO~\cite{azabou2023poyo}, and NEDS~\cite{zhang2025neds}. We evaluated a DeepSets baseline using the standard permutation-invariant form \(\rho(\sum_i \phi(\mathbf{x}_i))\), where each neuron is treated as one set element. Here, \(\mathbf{x}_i\) denotes the input feature vector for neuron \(i\), \(\phi\) is a
shared element encoder, and \(\rho\) is a set-level MLP applied after summation.
The input \(\mathbf{x}_i\) contains the neuron's activity segment, with the neuronal
identity vector concatenated when identity features are used. The element MLP
\(\phi\) produces neuron-level embeddings, which are summed across neurons and
passed through \(\rho\) to produce the decoding representation. We used the same
temporal input construction and task outputs as the other supervised decoding
baselines.

POYO~\cite{azabou2023poyo} and NEDS~\cite{zhang2025neds} were originally developed for spike-based neural recordings. In this work, we evaluate these baselines on datasets with different signal modalities, including calcium signals in the WARP dataset and spiking activity in the Allen datasets. For each dataset, we adapt only the input representation to match the signal modality, while preserving the original architectural design of each model as much as possible.

\paragraph{Calcium-based Input for the WARP dataset.}
For NEDS, we extend the input stage to process continuous calcium traces. Calcium activity is first divided into fixed-length temporal windows, yielding a neuron-by-time representation within each window. Each neuron's calcium signal is embedded into a latent space and concatenated with the corresponding neuron identity embedding. Since this representation includes both neuron and temporal dimensions, we use an attention-based adapter to aggregate neuron-level representations into temporally indexed tokens aligned with the input structure of NEDS. The resulting tokens are then passed to the original NEDS transformer encoder and task-specific decoding heads, yielding a calcium-adapted NEDS.

For POYO, we use a POYO+~\cite{azabou2025multisession} style calcium input representation. Specifically, neuron identity embeddings are added to the original unit embeddings, and calcium value embeddings at each time point are added to these representations. This allows POYO to process continuous calcium observations while retaining its token-based architecture, resulting in a calcium-adapted POYO.

\paragraph{Spike-based Input for the Allen Datasets.}
The Allen datasets consist of spike recordings, and thus both POYO and NEDS are used in their original spike-based formulations without modifying the input representation. 

To incorporate neuronal identity information, we extend each model in a manner consistent with its original design. For NEDS, neuron identity embeddings are combined with neural activity embeddings before being processed by the model. For POYO, identity embeddings are included with the unit embeddings from the original implementation. These unit embeddings were initialized following the official implementation. These modifications allow both models to incorporate neuronal identity information while preserving their original architectural assumptions.

\subsection{Hyperparameter Settings for Decoding Baselines}
\label{app:decoding_hyperparameters}

POYO is implemented as a transformer with cross-attention and self-attention, using a token dimension of 64, 6 layers, 2 cross-attention heads, and 8 self-attention heads. NEDS is implemented as a transformer encoder with a hidden size of 128, 3 layers, and 4 attention heads. In both models, neuronal identity features are encoded as 32-dimensional embeddings produced by an MLP comprising linear layers, an activation function, and layer normalization.

For training, the batch size is set to 64 across all experiments. POYO uses a learning rate of $3.125 \times 10^{-5}$ for both WARP and Allen datasets, while NEDS uses $3 \times 10^{-4}$ for WARP and $3 \times 10^{-5}$ for Allen.

\subsection{Manifold Baseline Implementation Details}
\label{app:manifold_baselines}

We evaluated CEBRA, MARBLE, joint PCA, and LFADS under the same held-out transfer protocol defined in Section~\ref{par:heldout-protocol}. Unless otherwise stated, all baseline embeddings were three-dimensional and were evaluated on the same held-out animals or sessions as the proposed model.

For CEBRA, we used the official CEBRA Python implementation. We trained a time-conditional CEBRA encoder on each source recording and applied the trained encoder directly to the held-out recording without refitting. We used the \texttt{offset1-model} architecture with \texttt{time\_offsets=1}, output dimension 3, learning rate \(3\times10^{-4}\), and temperature 1.0. WARP transfer runs used batch size 512 and 10,000 optimization iterations; Allen runs used batch size 256 and 2,000 iterations.

For MARBLE, we used the public MARBLE implementation. Neural activity was first normalized and projected to 10 principal components. We used the following setup: \(k=60\), 64 eigenvectors, and \texttt{inner\_product\_features=False}. MARBLE was trained for 160 epochs on the source recording. Held-out activity was transformed using the source PCA basis and the trained source MARBLE checkpoint.

For joint PCA, we concatenated source activity along time, fit PCA on the source recordings only, and applied the resulting projection to held-out activity. LFADS baselines used precomputed LFADS latent trajectories from the corresponding LFADS training environment; no LFADS parameters were refit during manifold comparison. For source-encoder baselines on Allen, metrics were averaged over source-session encoders.

\subsection{Manifold visualizations}

We additionally visualize latent trajectories from source and held-out fish across methods, colored by behavior (Fig.~\ref{figure_manifold_behavior}) and stimulus (Fig.~\ref{figure_manifold_stimulus}) labels.

\begin{figure}[t]
    \centering
    \centerline{\includegraphics[width=0.84\columnwidth]{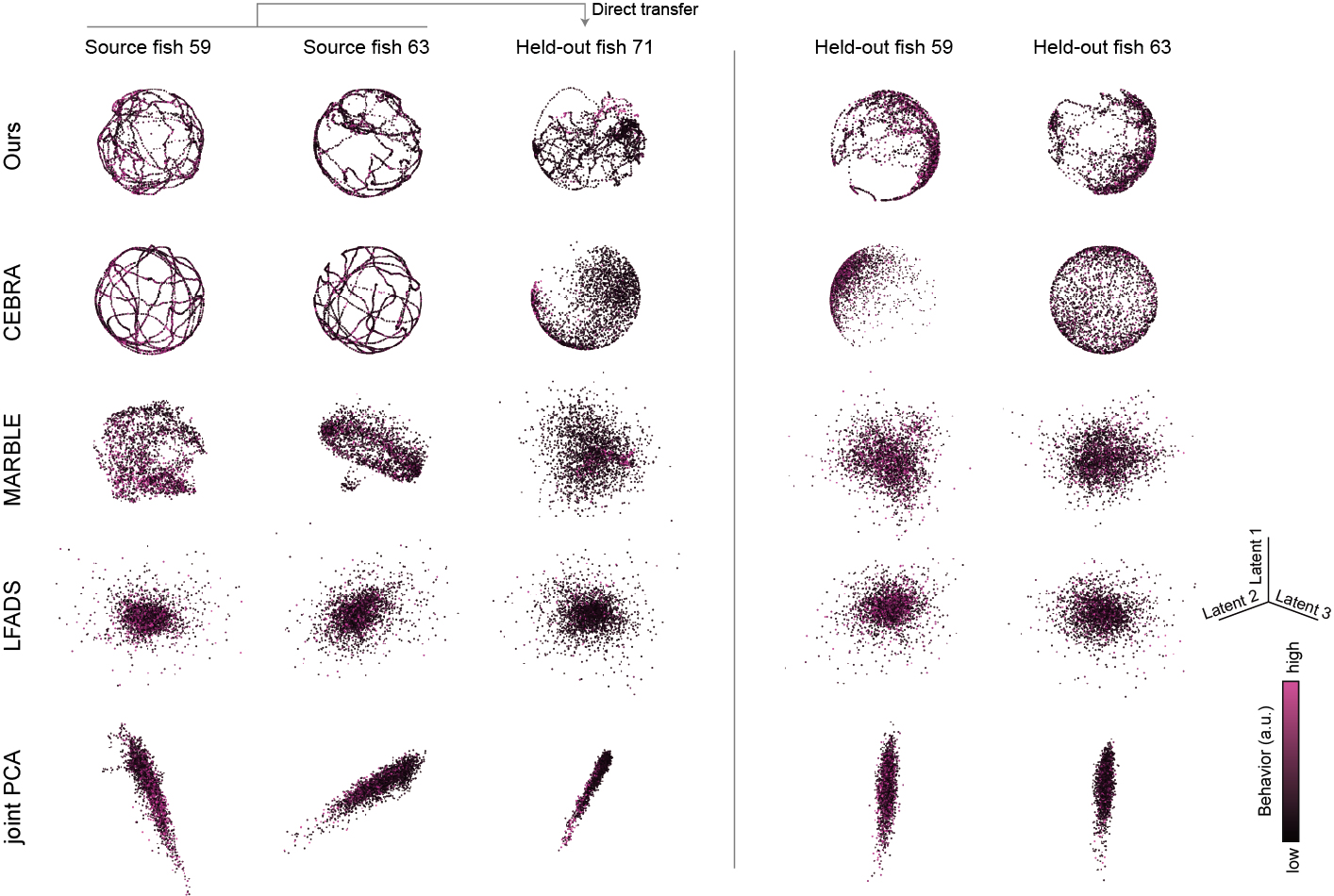}}
    \caption{Latent trajectories for two source fish (Fish 59 and Fish 63) and held-out fish across all methods. Points are colored by behavior label within each recording.
}
    \label{figure_manifold_behavior}
\end{figure}

\begin{figure}[t]
    \centering
    \centerline{\includegraphics[width=0.84\columnwidth]{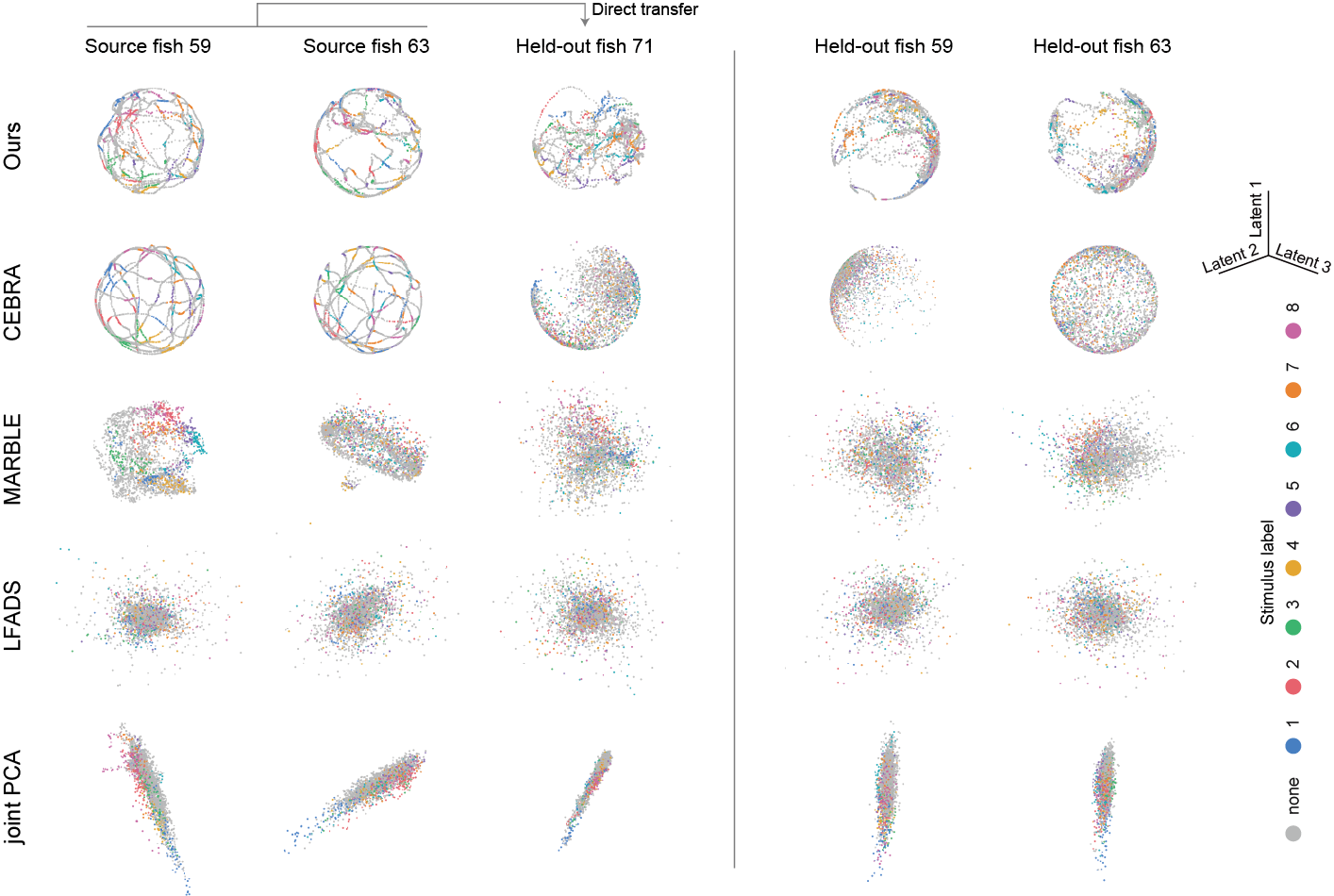}}
    \caption{Latent trajectories for two source fish (Fish 59 and Fish 63) and held-out fish across all methods. Points are colored by stimulus label within each recording.
}
    \label{figure_manifold_stimulus}
\end{figure}

\subsection{Effect of Neuronal Identity on Decoding Baselines}
\label{app:decoding_identity_effect}
To evaluate the contribution of neuronal identity, we incorporate multimodal biological features into baseline models. Details of the implementation and input construction are provided in Appendix~\ref{app:decoding_baseline_adaptation}. For each baseline, we compare the original model with a variant that incorporates per-neuron identity features. Training is performed only on source data, and models are evaluated on held-out recordings without any fine-tuning.

Baselines with identity features are distinguished based on the composition of multimodal biological features used. For the WARP dataset, we utilize transcriptomic profile (T), spatial position (S), HCR gene expression marker (H), anatomical region (R), and electrophysiological features (P), denoted as POYO + (T, S, H, R, P) and NEDS + (T, S, H, R, P). In contrast, for the Allen datasets, we use spatial position (S), anatomical region (R), and electrophysiological features (P), denoted as POYO + (S, R, P) and NEDS + (S, R, P).

As shown in Table \ref{tab:appendix_performancegain}, incorporating neuronal identity generally leads to improved baseline performance across datasets. Notably, performance gains are observed for NEDS and POYO even under the held-out transfer setting. This suggests that neuronal identity provides information that is useful for transfer without requiring retraining or neuron-wise correspondence.

\begin{table}[t!]
\centering
\caption{Behavior and stimulus decoding under direct held-out transfer across three datasets
(mean $\pm$ std). * denotes models adapted to calcium signals.}
\label{tab:appendix_performancegain}
\smallskip
\resizebox{0.9\linewidth}{!}{%
\begin{tabular}{llcc}
\toprule
Dataset & Method
    & Behavior (PCC) $\uparrow$
    & Stimulus (Balanced accuracy) $\uparrow$ \\
\midrule
\multirow{6}{*}{WARP}
 & {POYO* $+$ (T,S,H,R,P)}
    & { ${0.489 \pm 0.096}$} & {${0.164 \pm 0.042}$} \\
 & {POYO*~\cite{azabou2023poyo}}
    & {${0.304 \pm 0.120}$} & {${0.153 \pm 0.032}$} \\
 & {NEDS* $+$ (T,S,H,R,P)}
    & {${0.574 \pm 0.108}$} & {${0.374 \pm 0.074}$} \\
 & {NEDS*~\cite{zhang2025neds}}
    & {${0.248 \pm 0.207}$} & {${0.195 \pm 0.034}$} \\
 & DeepSets $+$ (T,S,H,R,P)
    & $0.466 \pm 0.051$
    & $0.163 \pm 0.023$ \\
 & DeepSets~\cite{zaheer2017deep}
    & $0.160 \pm 0.081$
    & $0.161 \pm 0.011$ \\
\midrule
\multirow{6}{*}{\shortstack[l]{Allen Brain\\Observatory}}
 & {POYO $+$ (S,R,P)}
    & {${0.695 \pm 0.003}$} & {${0.439 \pm 0.051}$} \\
 & {POYO~\cite{azabou2023poyo}}
    & {${0.526 \pm 0.043}$} & {${0.220 \pm 0.022}$} \\
 & {NEDS $+$ (S,R,P)}
    & {${0.498 \pm 0.055}$} & {${0.475 \pm 0.025}$} \\
 & {NEDS~\cite{zhang2025neds}}
    & {${0.500 \pm 0.042}$} & {${0.424 \pm 0.043}$} \\
 & DeepSets $+$ (S,R,P)
    & $0.623 \pm 0.060$
    & $0.727 \pm 0.063$ \\
 & DeepSets~\cite{zaheer2017deep}
    & $0.615 \pm 0.022$
    & $0.684 \pm 0.027$ \\
\midrule
\multirow{6}{*}{\shortstack[l]{Allen Functional\\Connectivity}}
 & {POYO $+$ (S,R,P)}
    & {${0.622 \pm 0.046}$} & {${0.494 \pm 0.068}$} \\
 & {POYO~\cite{azabou2023poyo}}
    & {${0.554 \pm 0.080}$} & {${0.211 \pm 0.008}$} \\
 & {NEDS $+$ (S,R,P)}
    & {${0.470 \pm 0.029}$} & {${0.449 \pm 0.043}$} \\
 & {NEDS~\cite{zhang2025neds}}
    & {${0.454 \pm 0.036}$} & {${0.398 \pm 0.020}$} \\
 & DeepSets $+$ (S,R,P)
    & $0.565 \pm 0.089$
    & $0.708 \pm 0.065$ \\
 & DeepSets~\cite{zaheer2017deep}
    & $0.588 \pm 0.096$
    & $0.730 \pm 0.044$ \\
\bottomrule
\end{tabular}
}
\end{table}

\subsection{Ablation Studies}
\label{app:ablations}

We conduct ablations to separate the contribution of identity
information from the contribution of the group-routing architecture.
Throughout this section, identity modalities are denoted by
\(T\) for transcriptomic profile, \(S\) for spatial position,
\(H\) for atlas-derived HCR gene markers, \(R\) for anatomical region, and
\(P\) for electrophysiological features.

\paragraph{Identity-modality ablations.}

Table~\ref{tab:ablation_warp} reports WARP modality ablations, removing
one modality at a time from the complete neuronal descriptor set \((T,S,H,R,P)\).
Across ablations, performance remains comparable to the all-modality configuration for most removals, indicating that the framework uses complementary identity features without depending on a single descriptor group.

\begin{table}[h]
\small
\centering
\caption{Modality ablation on WARP
(mean \(\pm\) std over 5 seeds). Each row removes one modality from the
full identity set. We denote static identity features by $T$ for transcriptomic profile, $S$ for spatial position, $H$ for atlas-derived HCR gene markers, and $R$ for anatomical region. We denote the electrophysiological features by $P$.}
\label{tab:ablation_warp}
\smallskip
\begin{tabular}{lcc}
\toprule
Method & Behavior (PCC) \(\uparrow\) & Stimulus (Balanced accuracy) \(\uparrow\) \\
\midrule
Ours (T,S,H,R,P)
    & \(0.793 \pm 0.015\)
    & \(0.835 \pm 0.029\) \\
Ours (S,H,R,P), w/o \(T\)
    & \(0.727 \pm 0.013\)
    & \(0.711 \pm 0.016\) \\
Ours (T,H,R,P), w/o \(S\)
    & \(0.721 \pm 0.014\)
    & \(0.573 \pm 0.010\) \\
Ours (T,S,R,P), w/o \(H\)
    & \(0.717 \pm 0.029\)
    & \(0.693 \pm 0.032\) \\
Ours (T,S,H,P), w/o \(R\)
    & \(0.716 \pm 0.020\)
    & \(0.682 \pm 0.030\) \\
Ours (T,S,H,R), w/o \(P\)
    & \(0.704 \pm 0.026\)
    & \(0.722 \pm 0.025\) \\
\bottomrule
\end{tabular}
\end{table}

Table~\ref{tab:ablation_allen} reports the corresponding Allen Brain
Observatory modality ablations using the available identity set
\((S,R,P)\). 
Both behavior decoding and stimulus classification remain effective under single-modality removals in this dataset.
We interpret the Allen ablation conservatively:
the compact \((S,R,P)\) identity set is sufficient for transfer when
molecular identity features are unavailable, while no single modality
alone explains the full behavior-decoding performance.

\begin{table}[h]
\small
\centering
\caption{Modality ablation on the Allen Brain Observatory dataset
(mean \(\pm\) std over 3 seeds). Each row removes one
modality from the full identity set. Here \(S\) denotes spatial position,
\(R\) anatomical region, and \(P\) electrophysiology features.}
\label{tab:ablation_allen}
\smallskip
\begin{tabular}{lcc}
\toprule
Method & Behavior (PCC) \(\uparrow\) & Stimulus (Balanced accuracy) \(\uparrow\) \\
\midrule
Ours (S,R,P)
    & \(0.734 \pm 0.008\)
    & \(0.893 \pm 0.012\) \\
Ours (R,P), w/o \(S\)
    & \(0.686 \pm 0.037\)
    & \(0.786 \pm 0.066\) \\
Ours (S,P), w/o \(R\)
    & \(0.711 \pm 0.007\)
    & \(0.800 \pm 0.040\) \\
Ours (S,R), w/o \(P\)
    & \(0.696 \pm 0.004\)
    & \(0.787 \pm 0.007\) \\
\bottomrule
\end{tabular}
\end{table}

\paragraph{Routing-design ablations.}

Table~\ref{tab:ablation_arch} isolates an architectural choice in
the group-routing framework on behavior decoding. Replacing soft routing with
hard top-1 assignment reduces performance, consistent with the need to represent
soft identity-to-group assignments rather than forcing each neuron into a single
discrete group. 

\begin{table}[h]
\centering
\caption{Architectural ablations on behavior decoding
(mean \(\pm\) std over 5 seeds for WARP and 3 seeds for Allen). Each row varies
one design choice from the full model. Bold indicates the best result within
each dataset.}
\label{tab:ablation_arch}
\smallskip
\resizebox{0.8\linewidth}{!}{%
\begin{tabular}{llcc}
\toprule
Dataset & Configuration & Changed component & Behavior (PCC) \(\uparrow\) \\
\midrule
\multirow{2}{*}{WARP}
 & \textbf{Full model}
    & Soft routing
    & \(\mathbf{0.793 \pm 0.015}\) \\
 & Hard top-1 routing
    & Soft routing removed
    & \(0.708 \pm 0.022\) \\
\midrule
\multirow{2}{*}{\shortstack[l]{Allen Brain\\Observatory}}
 & \textbf{Full model}
    & Soft routing
    & \(\mathbf{0.734 \pm 0.008}\) \\
 & Hard top-1 routing
    & Soft routing removed
    & \(0.678 \pm 0.003\) \\
\midrule
\multirow{2}{*}{\shortstack[l]{Allen Functional\\Connectivity}}
 & \textbf{Full model}
    & Soft routing
    & \(\mathbf{0.694 \pm 0.058}\) \\
 & Hard top-1 routing
    & Soft routing removed
    & \(0.626 \pm 0.078\) \\
\bottomrule
\end{tabular}
}
\end{table}

\subsection{Best-Ordering Analysis for Manifold Baselines}
\label{app:baseline_ordering}

Several manifold baselines operate on fixed input coordinates: their learned
encoders, projection matrices, or latent mappings are tied to the neuron indices
used during source training. In held-out transfer across animals, the target
recording contains a different set of neurons, and no neuron-wise correspondence
to the source recording is available. As a result, applying these baselines to a
held-out animal requires assigning an arbitrary order to the target neurons.

To make this comparison more favorable to the index-dependent baselines, we do
not rely on a single target-neuron ordering. Instead, for each held-out recording
and each baseline, we sample \(1000\) random permutations of the target neuron
order, apply the source-trained baseline under each ordering, and recompute the
three geometric trajectory metrics. We then examine the full distribution of
metric values as well as the best value obtained across sampled orderings.

Figures~\ref{fig:perm_temporal_corr}--\ref{fig:perm_local_linearity} show the
resulting distributions. Across baselines and metrics, our method generally outperforms the best sampled ordering.
This indicates that the observed gap is not due to choosing a single unfavorable
target ordering for the baselines. Rather, the index-dependent baselines lack a
mechanism for mapping an unaligned held-out population into a stable shared
coordinate system.

\begin{figure}[h]
    \centering
    \includegraphics[width=0.8\linewidth]{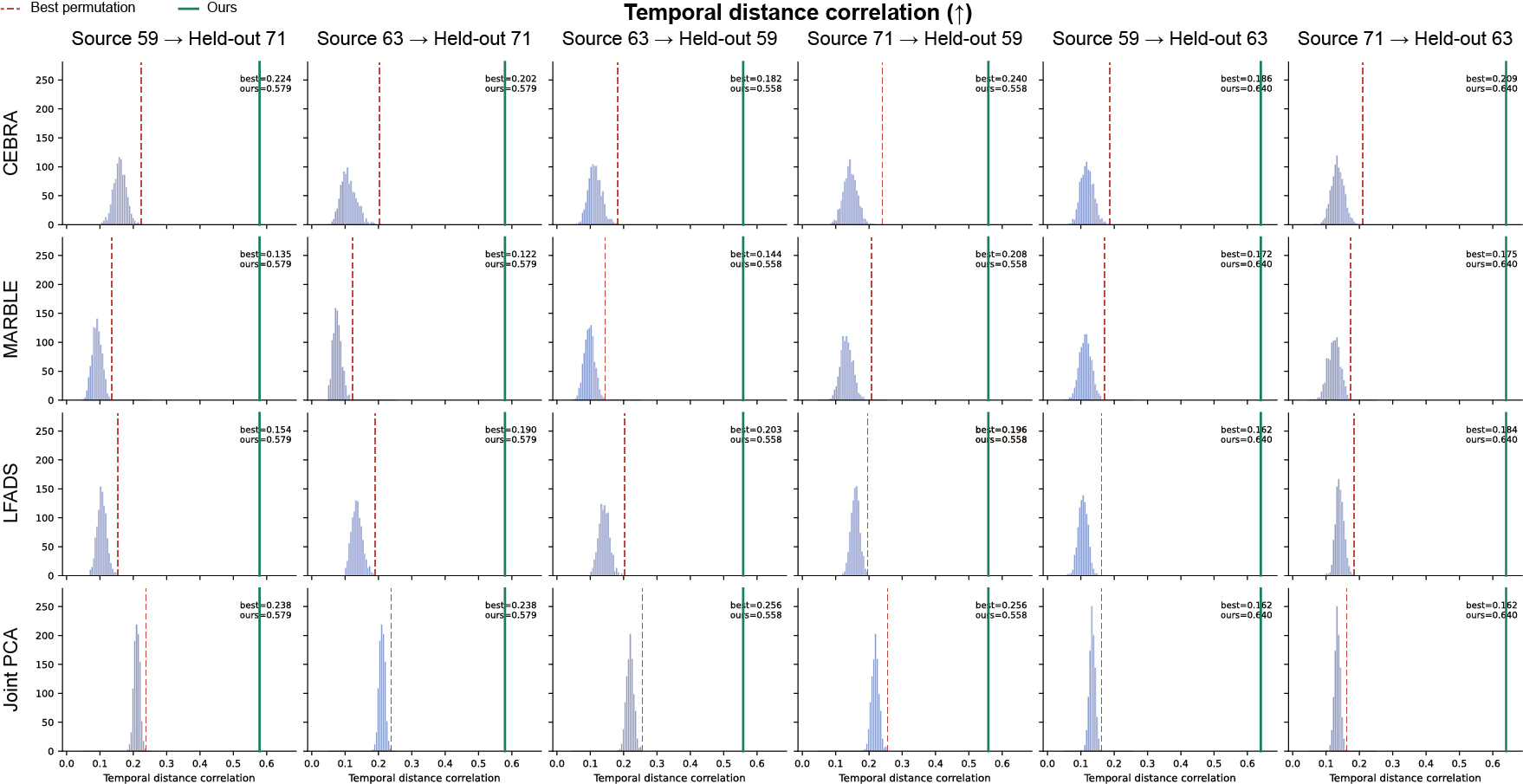}
    \caption{Target-ordering analysis for temporal distance correlation. Each
histogram shows baseline performance over \(1000\) sampled target-neuron
orderings. The vertical green line indicates our method and the vertical red dash indicates the best performance among the sampled orderings.}
    \label{fig:perm_temporal_corr}
\end{figure}

\begin{figure}[h]
    \centering
    \includegraphics[width=0.8\linewidth]{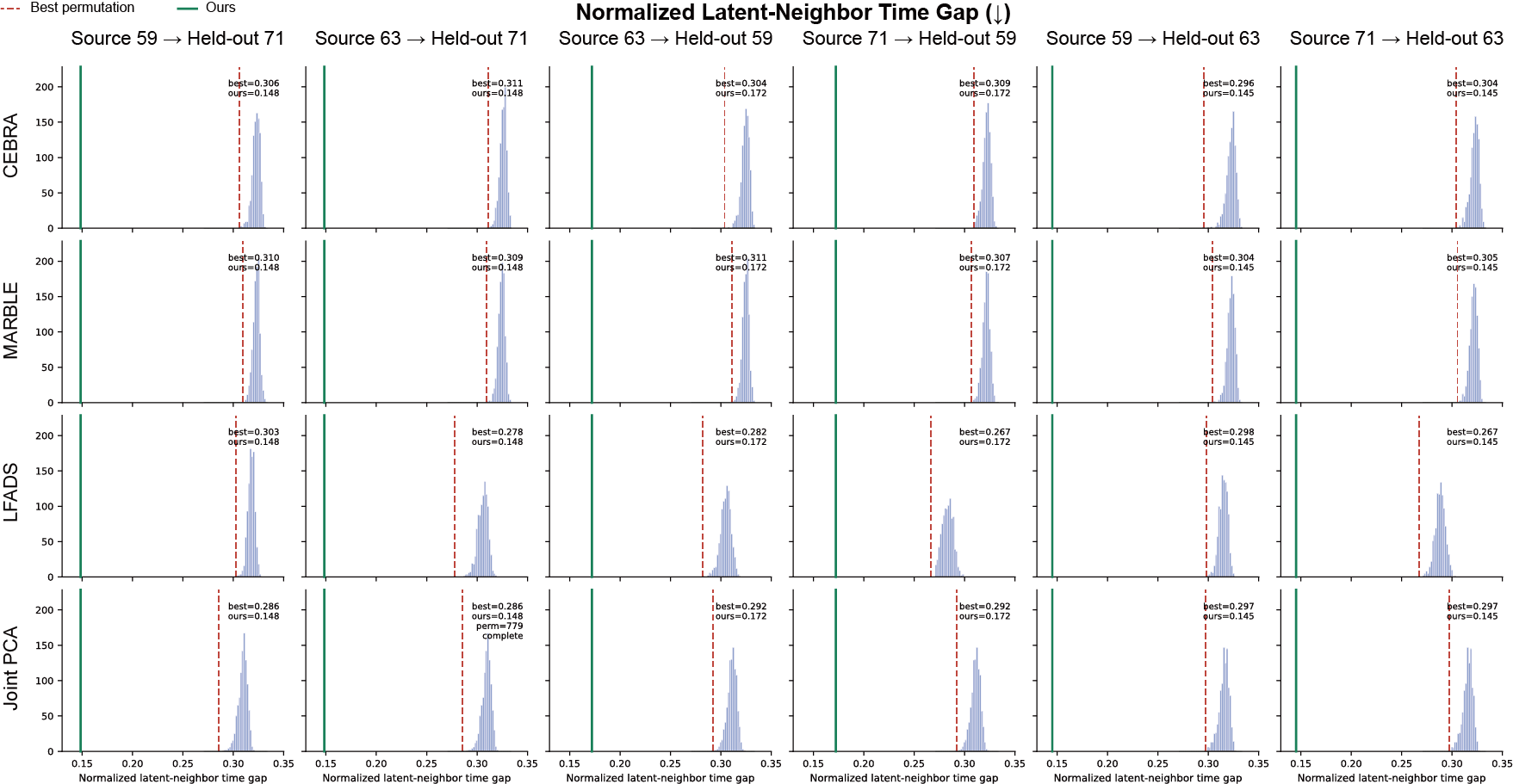}
    \caption{Target-ordering analysis for normalized latent-neighbor time gap. Each
histogram shows baseline performance over \(1000\) sampled target-neuron
orderings. The vertical green line indicates our method and the vertical red dash indicates the best performance among the sampled orderings.}
    \label{fig:perm_neighbor_gap}
\end{figure}

\begin{figure}[h]
    \centering
    \includegraphics[width=0.8\linewidth]{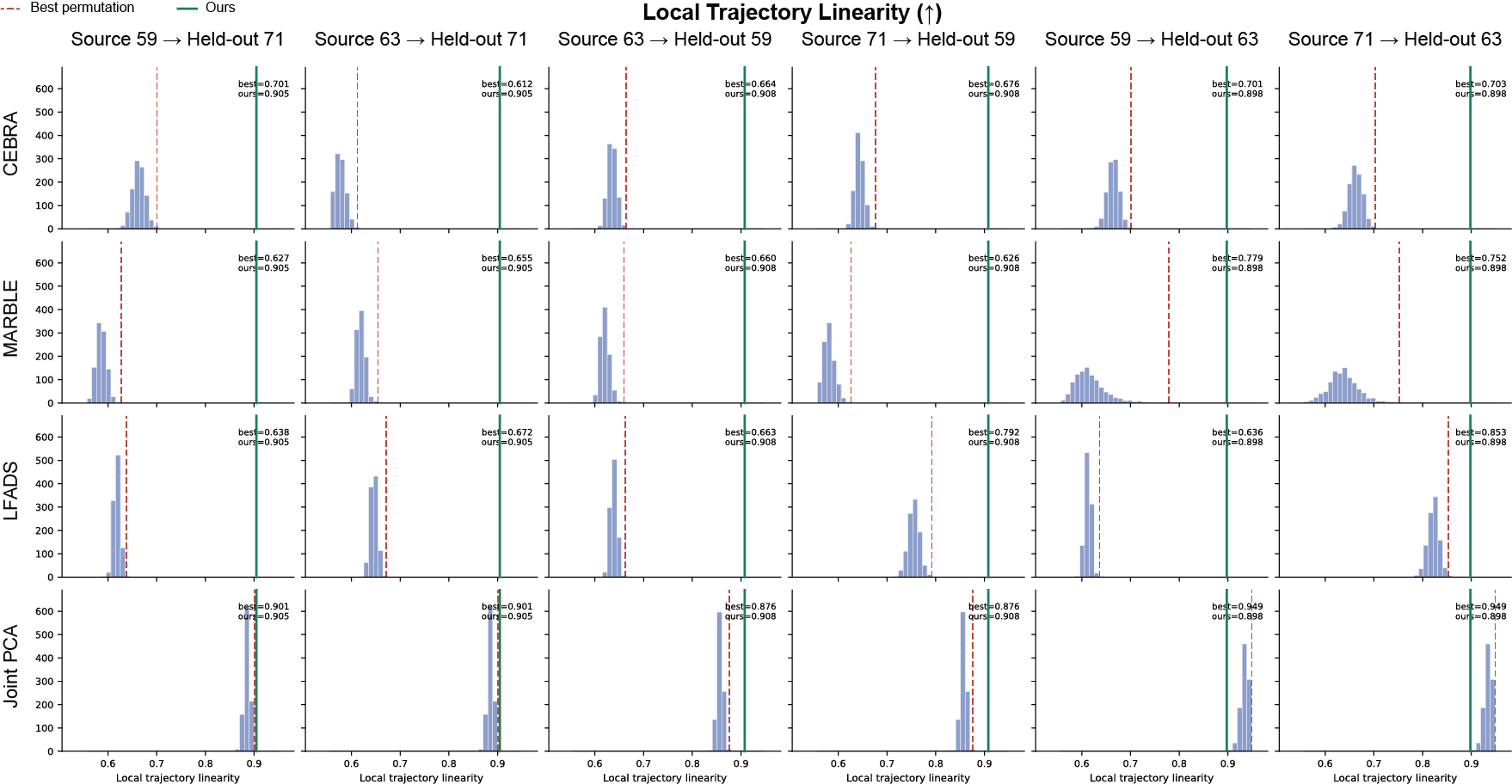}
    \caption{Target-ordering analysis for local trajectory linearity. Each histogram
shows baseline performance over \(1000\) sampled target-neuron orderings.  The vertical green line indicates our method and the vertical red dash indicates the best performance among the sampled orderings.}
    \label{fig:perm_local_linearity}
\end{figure}

\clearpage
\subsection{Compute Resources}
\label{app:compute}

All experiments were run on a Linux workstation with two Intel Xeon Silver 4214R CPUs, 440 GB system RAM, and three NVIDIA GeForce RTX 3090 GPUs with 24 GB memory per GPU. Each training or evaluation job used a single GPU.  All reported experiments fit within the memory of a single RTX 3090 GPU.

Table~\ref{tab:compute_resources} reports the wall-clock time required to reproduce the main reported experiments using our run configurations.

\begin{table}[h]
\centering
\caption{Compute required to reproduce the reported experiments. GPU-hours are estimated as wall-clock time multiplied by the number of GPUs used by the launcher.}
\label{tab:compute_resources}
\resizebox{0.9\linewidth}{!}{%
\begin{tabular}{lccc}
\toprule
Experiment group & GPUs used & Wall-clock time & Approx. GPU-hours \\
\midrule
WARP supervised decoding experiments & 3 RTX 3090 GPUs & 1.5 h & 4.6 \\
Allen supervised decoding experiments & 3 RTX 3090 GPUs & 7.8 h & 23.3 \\
WARP latent trajectory experiments & 3 RTX 3090 GPUs & 0.5 h & 1.5 \\
Allen latent trajectory experiments & 3 RTX 3090 GPUs & 6.4 h & 19.2 \\
\midrule
Total for reported experiments & -- & -- & $\approx 49$ GPU-hours \\
\bottomrule
\end{tabular}
}
\end{table}

The reported compute covers the runs used for the main decoding and latent trajectory tables, including the proposed method, ablations, and evaluated baselines. Additional exploratory runs used during development are not required to reproduce the reported results and are not included in Table~\ref{tab:compute_resources}.

\newpage

\end{document}